\documentclass{lmcs}

\usepackage{amsmath,amssymb,mathtools}
\usepackage{booktabs,array,tabularx}
\usepackage[hidelinks]{hyperref}

\makeatletter
\def\endfront@text{}
\def\enddoc@text{}
\makeatother

\hypersetup{pdftitle={Finite-Context Semantics in Finitely Supported Structures},pdfauthor={Gabriel Ciobanu},pdfkeywords={finitely supported structures, finite-context semantics, fixed points, data symmetries}}

\theoremstyle{plain}

\newtheorem{proposition}[thm]{Proposition}
\newtheorem{lemma}[thm]{Lemma}
\newtheorem{corollary}[thm]{Corollary}
\theoremstyle{definition}
\newtheorem{definition}[thm]{Definition}
\newtheorem{example}[thm]{Example}
\newtheorem{remark}[thm]{Remark}

\newcommand{\llbracket}{\mathopen{[\![}}
\newcommand{\rrbracket}{\mathclose{]\!]}}
\newcommand{\A}{\mathbb A}

\newcommand{\Pfs}{\mathcal P_{\mathrm{fs}}}
\newcommand{\PS}{\mathcal P_{S}}
\newcommand{\Pred}{\operatorname{Pred}}
\newcommand{\Post}{\operatorname{Post}}
\newcommand{\Pre}{\operatorname{Pre}}
\newcommand{\Ner}{\mathrel{\equiv}}

\newcommand{\Aut}{\operatorname{Aut}}
\newcommand{\Sym}{\operatorname{Sym}}
\newcommand{\im}{\operatorname{im}}

\newcommand{\lfp}{\operatorname{lfp}}
\newcommand{\gfp}{\operatorname{gfp}}

\begin{document}

\title[Finite-Context Semantics in Finitely Supported Structures]{Finite-Context Semantics in Finitely Supported Structures}

\author[Gabriel Ciobanu]{Gabriel Ciobanu \lmcsorcid{0000-0002-8166-9456}}
\address{Romanian Academy, IIT, Ia\c{s}i, Romania}
\email{gabriel.ciobanu@iit.academiaromana-is.ro}

\ACMCCS{Theory of computation~Semantics and reasoning;
Theory of computation~Automata over infinite objects;
Theory of computation~Logic and verification}


\begin{abstract}
Semantics for systems with names and data often depends on finitely many distinguished values. The theory of finitely supported structures treats this dependence through invariance under permutations fixing a finite context. We develop this approach to semantics over arbitrary permutation groups and arbitrary infinite sets of atoms. The central difficulty is that finitely supported predicate spaces need not be complete lattices. We prove that, for predicates valued in a complete lattice with trivial atom action, every monotone finitely supported transformer nevertheless has least and greatest fixed points. These lie in the complete lattice determined by the transformer's context and coincide with the fixed points of every compatible monotone ambient extension equivariant under its stabilizer. Support-transfer bounds track dependencies through semantic constructions, while uniform finiteness yields finite convergence. For Boolean predicates, $m$ context-stabilizer orbits on the carrier suffice for convergence after at most $m$ iterations, even when the supported predicate lattice is orbit-infinite. We apply these results to automata, operational and modal semantics, abstract interpretation, and resource rewriting. Ultrahomogeneous atom structures in finite relational signatures yield finite cell representations, illustrated by an authorization monitor. The resulting account separates semantic existence, finite convergence, and effective computation.
\end{abstract}

\maketitle

\section{Introduction}

Programs that manipulate names, locations, channels, registers, database values, or symbolic resources operate over domains that are naturally infinite, although each concrete computation usually depends on only finitely many distinguished values. Nominal techniques express this through symmetry: the relevant constructions commute with renamings of atoms, and every object is required to have a finite support \cite{Pitts2013}. The approach is particularly effective for binding, freshness, and $\alpha$-equivalence. Automata-theoretic approaches to infinite alphabets include finite-memory/register automata and data-word logics \cite{KaminskiFrancez1994,Segoufin2006}, symbolic automata~\cite{VeanesBjorner2012}, and automata over data symmetries, in which the group of admissible renamings is itself a parameter \cite{BojanczykKlinLasota2014}.

Many computational semantics are only locally equivariant: they distinguish a finite environment while respecting the data symmetries that fix that environment. A transition rule may depend on the registers currently in scope; a specification may mention selected channels or allocated locations; a security analysis may distinguish finitely many principals. Over a data symmetry $(\A,G)$ with $G\leq\Sym(\A)$, such an object is supported by a finite set $S\subseteq\A$, its \emph{context}: it is invariant under the pointwise stabilizer $G_{(S)}$. The theory of finitely supported structures (FSS) \cite{AlexandruCiobanu2020} treats context-dependent subsets, maps, relations, and semantic transformers of this kind as mathematical objects in their own right, and it studies which classical constructions survive when every object involved, including every intermediate object of a proof, must be finitely supported. This paper brings that viewpoint to program semantics and verification.

\paragraph{Scope and contribution}
The starting point is the FSS programme of Alexandru and Ciobanu: internal completeness, fixed-point theorems, and distinctions between forms of finiteness are developed in \cite{AlexandruCiobanu2020}, with applications to resourse algebra, abstract interpretation, and fuzzy predicates in \cite{AlexandruCiobanu2015Multisets,AlexandruCiobanu2016Abstract,AlexandruCiobanu2018Fuzzy}. We organize these constructions by an explicit finite context and prove the required support, fixed-point, and termination results for arbitrary data symmetries. The main connections are the coincidence of internal, context, and compatible ambient fixed points, and the passage from uniform finiteness to explicit iteration bounds and finite semantic representations.

At a fixed context $S$, supported maps are equivariant for $G_{(S)}$ (Remark~\ref{rem:context-symmetry}), so this formulation is compatible with nominal methods over data symmetries \cite{BojanczykKlinLasota2014}. Standard support closure, inductive definitions, and modal support arguments retain their established attributions \cite{Pitts2006,KlinLelyk2019}. The contribution is a self-contained development connecting these results through FSS, rather than a claim that each constituent closure theorem is new. The paper concerns constructions that preserve a declared context. The separate question of how many parameters must be added to select an implementation is outside its scope.

\paragraph{Three structural properties}
Three questions arise whenever a semantics depends on a finite context, and they are easily conflated. Each is answered by one structural property, stated and proved in Section~\ref{sec:principles}.
\begin{enumerate}[label=(\Roman*)]
\item \emph{Transfer: how is the context of a composite object computed, and how does it grow?}  Supports of composite objects are bounded by the supports of their parts and operations (Propositions~\ref{prop:criterion}, \ref{thm:term-support}, and~\ref{prop:finite-store}), and deterministic evaluation introduces no support dependencies beyond those of its inputs and operations (Proposition~\ref{prop:no-fresh-generation}).
\item \emph{Internal completeness: where do fixed points live?}  The finitely supported powerset need not be a complete lattice, but it is complete for finitely supported families, which are the only families a finitely supported construction produces (Proposition~\ref{prop:internal-completeness}). Every finitely supported monotone transformer on these predicate spaces therefore has least and greatest fixed points; they lie in the complete lattice of predicates supported by the context of the transformer and coincide with those of every compatible monotone ambient extension equivariant under its stabilizer (Proposition~\ref{prop:fs-fixed-points}).
\item \emph{Uniform finiteness: when is a supported semantics finite?}  For Boolean predicates, if the stabilizer of the context has $m$ orbits on the carrier, fixed-point iteration stabilizes after at most $m$ steps, even when the supported predicate lattice has infinitely many orbits (Proposition~\ref{thm:terminating-iteration} and Corollary~\ref{cor:iteration-bound}); for oligomorphic symmetries the orbit hypothesis holds at every context on finite powers of atoms and their invariant subsets (Lemma~\ref{lem:local-oligomorphic}).
\end{enumerate}
Semantic existence needs no orbit-finiteness hypothesis; finite context orbits ensure termination for Boolean predicates (and finite-height quantitative lattices); effectiveness additionally needs computable operations on an effective orbit presentation (Section~\ref{sec:structured}). For the equality symmetry, Properties~II and~III are results of the monograph \cite[Chapters~5--7]{AlexandruCiobanu2020}; here they are proved for every data symmetry, the internal fixed points are identified with the context and ambient ones, and the iteration bound is made explicit. The features of FSS that these properties rest on (arbitrary atom sets, the distinction between internal and ambient choice, internal completeness, and graded notions of finiteness) are summarized in Section~\ref{sec:fss-viewpoint}.

\paragraph{Applications}
Each application section records the property it uses.
\begin{enumerate}[label=(\roman*)]
\item Determinization on the finitely supported powerset and context-relative Nerode quotients, without orbit-finiteness; finite algorithms under a separate orbit hypothesis (Sections~\ref{sec:automata} and~\ref{sec:quotients}).
\item Support of rule-generated transition relations and strong bisimilarity as a greatest fixed point in the lattice of $S$-supported relations, with an explicit iteration bound (Section~\ref{sec:concurrency}).
\item Modal fixed-point semantics, quantitative predicate lattices, Galois connections, rough approximations, and widening/narrowing at a fixed context (Sections~\ref{sec:logic} and~\ref{sec:ai}).
\item Free resource monoids and multiset rewriting over possibly orbit-infinite alphabets (Section~\ref{sec:resources}).
\item Finite cell presentations for homogeneous structures, with explicit counts for equality atoms, dense order, and the random graph, and a worked authorization monitor (Section~\ref{sec:structured}).
\end{enumerate}
Section~\ref{sec:frontier} studies the hypotheses and property used by each construction, and Section~\ref{sec:directions} outlines research directions organized by property.

\paragraph{Reading guide}
Section~\ref{sec:principles} contains the common mathematical core. Readers interested in finite representations can continue directly to Section~\ref{sec:structured}; the intervening application sections can be read independently. Their purpose is to expose the support and finiteness hypotheses of familiar constructions, not to present separate new algorithms.

\paragraph{A modelling example}
In the standard nominal category, transition maps are globally equivariant. A policy-dependent transition on a fixed finite control set, such as a monitor that accepts a login by a trusted principal $p$ followed by access to a protected resource $r$, is not such a morphism on the unchanged carriers (Proposition~\ref{prop:authorization-no-global}). It is represented nominally by threading $(p,r)$ through the state, by naming $p$ and $r$ as constants, or by working in the slice over the orbit of $(p,r)$. In the fixed-context presentation the control set is unchanged and the pair $(p,r)$ appears once, as the support of the transition map. The presentations are equivalent; the difference is one of bookkeeping, not of expressive power. In the terminology of the invariance spectrum of \cite{Ciobanu2026PhilMath}, such a transition map is not globally equivariant but context-relatively structural: it is invariant under the stabilizer of its finite~context.

\section{Data symmetries, contexts, and the finitely supported viewpoint}\label{sec:prelim}

We work over a \emph{data symmetry} in the sense of \cite{BojanczykKlinLasota2014}: a pair $(\A,G)$ consisting of an infinite set $\A$ of atoms and a permutation group $G\leq\Sym(\A)$. The group records which structure of the data is observable. Standard examples are the equality symmetry, where $G$ is the group of all (or of all finitary) permutations of $\A$; the ordered symmetry, where $\A=\mathbb Q$ and $G=\Aut(\mathbb Q,<)$; and the random-graph symmetry, where $G$ is the automorphism group of the countable random graph on $\A$. For $S\subseteq\A$ put $G_{(S)}=\{g\in G:g(a)=a\text{ for every }a\in S\}$. A \emph{context} is a finite set $S\subseteq\A$; it does not change the data symmetry, but selects the atoms on which a particular object may depend. We use the term \emph{finite-context dependency analysis} for the task of computing a context that supports a given expression, transition system, specification, or semantic transformer.

\paragraph{Notation}
For a group $H$ acting on a set $X$, we write $H\backslash X$ for the set of $H$-orbits, and $X^H=\{x\in X:hx=x\text{ for all }h\in H\}$ for the set of $H$-fixed elements. We write $[\A]^{<\omega}$ for the set of finite subsets of $\A$. A group acts on subsets by $gU=\{gx:x\in U\}$, on maps $f:X\to Y$ by conjugation, $(g\cdot f)(x)=g f(g^{-1}x)$, and componentwise on tuples and words; relations are regarded as subsets of products.

\begin{definition}[Support]
Let $X$ be a $G$-set. A set $S\subseteq\A$ \emph{supports} $x\in X$ if $gx=x$ for all $g\in G_{(S)}$. An element is \emph{finitely supported} if some finite set supports it, and $X$ is a \emph{nominal $G$-set} if all its elements are finitely supported. For a context $S$, an element supported by $S$ is called \emph{$S$-supported}.
\end{definition}

Nominal $G$-sets are the objects studied in \cite{BojanczykKlinLasota2014}; for the equality symmetry they are the nominal sets of \cite{Pitts2013}, and in the FSS literature they are called invariant sets \cite{AlexandruCiobanu2020}. For a nominal $G$-set $X$ and a context $S$ define
\begin{gather*}
\Pfs(X)=\{U\subseteq X:U\text{ is finitely supported}\},\\
\PS(X)=\{U\subseteq X:gU=U\text{ for all }g\in G_{(S)}\}.
\end{gather*}

\begin{proposition}[Restriction of the data symmetry]\label{prop:symmetry-refinement}
Let $H\leq G\leq\Sym(\A)$ and let $X$ be a $G$-set. If $S$ supports $x\in X$ with respect to $G$, then $S$ supports $x$ with respect to $H$. In particular, the restriction of a nominal $G$-set to $H$ is a nominal $H$-set.
\end{proposition}

\begin{proof}
$H_{(S)}=H\cap G_{(S)}\subseteq G_{(S)}$.
\end{proof}

Passing to a subgroup therefore preserves all support bounds while admitting more supported predicates and maps; for the hereditarily finitely supported universes $\mathrm{FS}_G(\A)$ this is the monotonicity $\mathrm{FS}_{G}(\A)\subseteq\mathrm{FS}_{H}(\A)$ of \cite[Proposition~26]{Ciobanu2026Synthese}. The case $H=G_{(S)}$ explains the role of contexts.

\begin{remark}[Contexts as data symmetries]\label{rem:context-symmetry}
Fix a context $S$. Since $(G_{(S)})_{(T)}=G_{(S\cup T)}$, a finite set $T$ supports $x$ with respect to $G_{(S)}$ if and only if $S\cup T$ supports $x$ with respect to~$G$. Hence the $S$-supported elements of a nominal $G$-set are exactly its $G_{(S)}$-fixed elements, and, by Proposition~\ref{prop:criterion} below, the $S$-supported maps between nominal $G$-sets are exactly the equivariant maps between the corresponding nominal $(\A,G_{(S)})$-sets. A fixed context is therefore handled by the nominal theory of the data symmetry $(\A,G_{(S)})$. For the equality symmetry, $G_{(S)}$ fixes $S$ pointwise and acts on $\A\setminus S$ as the full (respectively, finitary) symmetric group, so $S$-supported structure is ordinary nominal structure over the atoms $\A\setminus S$, with the atoms of $S$ acting as constants. For the ordered symmetry, $G_{(S)}$ is the automorphism group of $(\mathbb Q,<)$ expanded by constants for the elements of $S$.

For the full category of nominal $G_{(S)}$-sets, let $\bar s$ enumerate $S$, so that the stabilizer of $\bar s$ in $G$ is $G_{(S)}$. Taking the fibre over $\bar s$ and, conversely, the induced set $G\times_{G_{(S)}}Y$ give the classical equivalence between $G$-sets over the orbit $G\bar s$ and $G_{(S)}$-sets. Both constructions preserve finite support, because $[g,y]$ is supported by $g(S\cup T)$ whenever $T$ supports $y$ with respect to $G_{(S)}$. Thus nominal $(\A,G_{(S)})$-sets are equivalent to nominal $G$-sets over $G\bar s$; this is the slice presentation of a context. This equivalence allows induced carriers; it does not assert that restriction on a prescribed collection of unchanged $G$-carriers is itself an equivalence with the full category of nominal $G_{(S)}$-sets.

The explicit parameter tracks dependencies across different contexts: the context of a composite object is computed from those of its parts (Propositions~\ref{prop:criterion} and~\ref{thm:term-support}), it grows along runs that read fresh data (Corollary~\ref{cor:reachable-support}), and it is the level at which classical completeness becomes available (Propositions~\ref{prop:orbit-algebra} and~\ref{prop:internal-completeness}). The slice presentation carries the same information in an enlarged carrier; see the discussion after Proposition~\ref{prop:authorization-no-global}.
\end{remark}

\subsection{The finitely supported viewpoint}\label{sec:fss-viewpoint}

The monograph \cite{AlexandruCiobanu2020} develops finitely supported mathematics as ordinary set theory in which every construction involving atoms must be finitely supported; a classical theorem is available there only once it has been reproved with finitely supported objects, and some classical theorems fail \cite[Sections~1.3--1.4 and Chapter~3]{AlexandruCiobanu2020}. The monograph is formulated for the group of all (finitary) permutations of $\A$. In \cite{Ciobanu2026Synthese} the group is a \emph{calibration parameter} of the finitely supported universe: it determines which choice, orderability, and cardinality principles hold internally, and, when $\A$ is well-orderable in the metatheory, the universe collapses to the whole cumulative hierarchy exactly when $G$ has a finite base \cite[Proposition~33]{Ciobanu2026Synthese}. Four features of the viewpoint are used below.

\paragraph{Arbitrary infinite atom sets}
No countability assumption is made on $\A$ \cite[Section~1.3]{AlexandruCiobanu2020}, and none of the general results of this paper uses one; countability enters only in examples such as the random graph. Proposition~\ref{thm:homogeneous-cells} applies equally to uncountable homogeneous structures such as $(\mathbb R,<)$, natural carriers for real-valued or timed data. Uncountable atom sets also separate finite support from syntax: no fixed countable supported first-order language with countably many parameters defines all finitely supported singletons \cite[Proposition~24 and Corollary~25]{Ciobanu2026Synthese}. Thus finitely supported semantic objects over an uncountable data domain need not all be expressible in one fixed countable specification language, although in the canonical settings each is definable over its own support (Section~\ref{sec:structured}).

\paragraph{Transfer with explicit supports}
The $S$-finite support principle \cite[Section~1.3]{AlexandruCiobanu2020}, corresponding to \cite[Theorem~3.5]{Pitts2006}, states that anything definable in higher-order logic from $S$-supported data by $S$-supported constructions is $S$-supported. We use its constructive form, the \emph{hierarchical construction of supports}: exhibit a finite context for each component and verify the resulting support bound directly. This method makes the dependencies of each semantic construction explicit.

\paragraph{Choice and the metatheory}
We distinguish ordinary reasoning about $G$-sets from choice inside a finitely supported universe. In the equality-atom setting, there is no finitely supported choice function on the family of all two-element subsets of $\A$: a transposition of a pair outside a proposed support contradicts the choice. This concerns a family of finite sets, not a finite family; finite choice remains valid. The monograph \cite[Chapter~3]{AlexandruCiobanu2020} studies failures of choice and related principles in FSS. Their pattern depends on the symmetry group \cite{Ciobanu2026Synthese}. For example, ordered atoms admit an invariant linear order, whereas equality atoms do not.

The semantic proofs below use no choice principle beyond finite choice, which is a theorem of ZF. They concern nominal $G$-sets and supported maps directly; they do not assume that every ambient subset is supported. The general representative presentation in Proposition~\ref{thm:map-representation} is conditional on a supplied transversal. Its finite-orbit instances need only finite choice. No claim is made that an arbitrary orbit transversal is itself finitely supported.

\paragraph{Internal completeness and graded finiteness}
The finitely supported powerset and the finitely supported lattice-valued functions are guaranteed to have joins and meets for finitely supported families \cite[Definition~6.2, Theorems~7.1 and~7.3]{AlexandruCiobanu2020}, and Tarski's theorem holds for them in that internal form \cite[Theorem~6.2]{AlexandruCiobanu2020}. Finiteness splits into inequivalent notions for finitely supported sets \cite[Chapter~9]{AlexandruCiobanu2020}, \cite{AlexandruCiobanu2022Infinity,AlexandruCiobanu2024Finite}; the one that governs termination of fixed-point iteration is the absence of infinite uniformly supported subsets \cite{AlexandruCiobanu2019Uniform,AlexandruCiobanu2020Carpathian}. These are Properties~II and~III.

\section{Three structural properties of finite-context semantics}\label{sec:principles}

This section states the three structural properties that the rest of the paper applies. Property~I concerns supports, Property~II concerns completeness and fixed points, and Property~III concerns termination. The statements specialize to the equality symmetry as indicated, where they are due to \cite{Pitts2006,Pitts2013,AlexandruCiobanu2020}; all proofs are given for an arbitrary data symmetry and use only the existence of some finite support, never least supports or finitary permutations. This matters because least supports are a special feature of particular symmetries and do not follow from oligomorphicity alone \cite[Lemma~8 and the discussion following it]{Ciobanu2026Synthese}.

\subsection{Property I: transfer of supports}\label{sec:principle-transfer}

The following facts are standard; see \cite{Pitts2013} for the equality symmetry. The proofs for an arbitrary data symmetry are the same.

\begin{proposition}[Finite-context criterion]\label{prop:criterion}
A map $f:X\to Y$ between $G$-sets is supported by~$S$ if and only if
$f(gx)=g f(x)$ for all $g\in G_{(S)}$ and $x\in X$.
If $x$ is supported by $T$, then $f(x)$ is supported by $S\cup T$. If $h:Y\to Z$ is supported by~$T$, then $h\circ f$ is supported by $S\cup T$.
\end{proposition}

\begin{proof}
The first assertion is obtained by expanding $g\cdot f=f$. If $k\in G_{(S\cup T)}$, then $kx=x$ and $k\cdot f=f$, hence $kf(x)=f(kx)=f(x)$. The composition statement follows by applying the first assertion twice.
\end{proof}

The empty context recovers global equivariance. A finite support identifies parameters that admissible renamings must fix for evaluation to commute with renaming. For the full symmetric group, Proposition~\ref{prop:criterion} is the Local Invariance Theorem of \cite[Theorem~8]{Ciobanu2026PhilMath}, where the passage from global equivariance to invariance under the stabilizer of a finite context is organized as an invariance spectrum governed by a Galois connection between symmetry groups and the operations they fix; the operations fixed by $\Sym(\A)_{(F)}$ are exactly those supported by $F$ \cite[Proposition~12]{Ciobanu2026PhilMath}. The context lattices $\PS(X)$ of Section~\ref{sec:principle-completeness} are the predicate-level counterpart of these local vocabularies. When the data symmetry admits least supports \cite{BojanczykKlinLasota2014}, this set has a canonical minimum.

\begin{proposition}[Support bound for algebraic expressions]\label{thm:term-support}
Let $\tau$ be a finitary signature and let $\mathcal A$ be a $G$-set equipped, for every $k$-ary operation symbol $\omega\in\tau$, with a finitely supported map $\omega^{\mathcal A}:\mathcal A^k\to\mathcal A$. Fix a finite term $t(x_1,\ldots,x_n)$. Choose a finite support $S_\omega$ for every primitive operation occurring in $t$ and a finite support $T_i$ for each value $a_i\in\mathcal A$. Then
\[
\bigcup_{\omega\text{ occurs in }t}S_\omega
\ \cup\ 
\bigcup_{i=1}^{n}T_i
\]
supports $t^{\mathcal A}(a_1,\ldots,a_n)$.
\end{proposition}

\begin{proof}
Induct on the construction of the term. Variables use the chosen supports $T_i$. At a composite term $\omega(t_1,\ldots,t_k)$, apply Proposition~\ref{prop:criterion} to $\omega^{\mathcal A}$ and the inductive support bounds of the subterms.
\end{proof}

Proposition~\ref{thm:term-support} is the basic soundness statement for finite-context dependency analysis, and an instance of the finite support principle: evaluating a finite expression cannot introduce dependence on atoms absent from its inputs and primitive operations. Supports of subexpressions are combined by union.

\begin{proposition}[Finite stores and environments]\label{prop:finite-store}
Let $V$ be a nominal $G$-set. The set $\operatorname{Env}_{\mathrm{fin}}(\A,V)$ of finite partial maps $\rho:\A\rightharpoonup V$, with action defined by
\[
\operatorname{dom}(g\rho)=g\operatorname{dom}(\rho),\qquad
(g\rho)(a)=g\bigl(\rho(g^{-1}a)\bigr)
\quad(a\in g\operatorname{dom}(\rho)),
\]
is a nominal $G$-set. If $D=\operatorname{dom}(\rho)$ and $T_a$ supports $\rho(a)$ for $a\in D$, then $D\cup\bigcup_{a\in D}T_a$
supports $\rho$. In particular, if $V$ has trivial action, the finite domain $D$ supports the environment.
\end{proposition}

\begin{proof}
The displayed set is finite. A permutation fixing it pointwise fixes every address in the domain and every stored value, hence fixes the graph of~$\rho$. Transport of graphs gives the stated group action.
\end{proof}

Finite stores, register files, typing environments, channel environments, and finite security policies are therefore finitely supported objects with explicit support bounds. The results below do not require least supports. When the chosen symmetry admits least supports, they provide canonical minimal contexts and may improve symbolic implementations; all semantic closure arguments in this paper use only the existence of some finite support.

The following consequence of Proposition~\ref{prop:criterion} bounds the outputs of a deterministic enumeration with a fixed context.

\begin{proposition}[A context does not generate fresh data]\label{prop:no-fresh-generation}
Let $X$ be a $G$-set and let $f:\mathbb N\to X$ be supported by $S$, where $\mathbb N$ carries the trivial action. Then every value $f(n)$ is $S$-supported. Consequently, if the set $X^{G_{(S)}}$ of $S$-supported elements of $X$ is finite, then $f$ has finite image and is not injective.
\end{proposition}

\begin{proof}
For $g\in G_{(S)}$, Proposition~\ref{prop:criterion} gives $gf(n)=f(gn)=f(n)$. Hence $f(\mathbb N)\subseteq X^{G_{(S)}}$.
\end{proof}

For $X=\A$ and the equality, ordered, and random-graph symmetries of Section~\ref{sec:structured}, $\A^{G_{(S)}}=S$: every atom outside $S$ is moved by some element of $G_{(S)}$. Thus $\A$ is infinite, yet no finitely supported map $\mathbb N\to\A$ is injective; for the equality symmetry this is \cite[Theorem~10.1(11)]{AlexandruCiobanu2020}, and it says that $\A$ is not Dedekind infinite in the finitely supported sense \cite[Chapter~9]{AlexandruCiobanu2020}. Operationally, a deterministic procedure fixed by a finite context and driven by an ordinary counter emits only atoms of its context. Fresh atoms enter a computation through its inputs, as in Corollary~\ref{cor:reachable-support} below, or through a nondeterministic choice from an orbit, as in the name-generation rules of process calculi.

\subsection{Property II: internal completeness and fixed points}\label{sec:principle-completeness}

Once a context is fixed, the subsets it supports form a complete Boolean algebra.

\begin{proposition}[Context lattices]\label{prop:orbit-algebra}
Let $X$ be a nominal $G$-set and $S$ a context. Then $\PS(X)=\mathcal P(X)^{G_{(S)}}$ is closed under arbitrary unions, arbitrary intersections, and complements, and
\[
\PS(X)\cong\mathcal P(G_{(S)}\backslash X),
\qquad
U\longmapsto\{O\in G_{(S)}\backslash X:O\subseteq U\},
\]
is an isomorphism of complete Boolean algebras. We call $\PS(X)$ the \emph{context lattice} of $S$ on~$X$. Moreover,
\[
\Pfs(X)=\bigcup_{S\in[\A]^{<\omega}}\PS(X),
\]
and this union is directed, since $\PS(X)\cup\mathcal P_T(X)\subseteq\mathcal P_{S\cup T}(X)$.
\end{proposition}

\begin{proof}
An $S$-supported subset is exactly a union of $G_{(S)}$-orbits. Sending such a union to the corresponding subset of the orbit set gives the Boolean isomorphism. Arbitrary unions and intersections of $G_{(S)}$-invariant subsets remain $G_{(S)}$-invariant.
\end{proof}

For a complete lattice $L$ with the trivial $G$-action, define
\[
\Pred_S(X,L)=\{p:X\to L:\; p(gx)=p(x) \text{for all }g\in G_{(S)},\ x\in X\},
\]
and let
\[
\Pred_{\mathrm{fs}}(X,L)=\bigcup_{T\in[\A]^{<\omega}}\Pred_T(X,L)
\]
be the $G$-set of finitely supported $L$-valued predicates, with the conjugation action $(gp)(x)=p(g^{-1}x)$ and the pointwise order. For $L=\{0,1\}$ these are $\PS(X)$ and $\Pfs(X)$. A family $\mathcal F\subseteq\Pfs(X)$ or $\mathcal F\subseteq\Pred_{\mathrm{fs}}(X,L)$ is \emph{supported by $S$} if $g\mathcal F=\mathcal F$ for all $g\in G_{(S)}$.

\begin{proposition}[Internal completeness]\label{prop:internal-completeness}
Let $X$ be a nominal $G$-set and $L$ a complete lattice with trivial action. If a family $\mathcal F\subseteq\Pred_{\mathrm{fs}}(X,L)$ is supported by $S$, then its pointwise join and meet belong to $\Pred_S(X,L)$, with the convention that the empty meet is the top predicate. In particular, the union and the intersection of an $S$-supported family of finitely supported subsets of $X$ belong to $\PS(X)$. Thus $\Pfs(X)$ and $\Pred_{\mathrm{fs}}(X,L)$ are complete with respect to finitely supported families, and at every fixed context $S$ the lattices $\PS(X)$ and $\Pred_S(X,L)$ are complete in the classical sense.
\end{proposition}

\begin{proof}
Let $g\in G_{(S)}$ and $x\in X$. Since $L$ has trivial action,
\[
\Bigl(g\bigvee\mathcal F\Bigr)(x)=\bigvee_{q\in\mathcal F}q(g^{-1}x)=\bigvee_{q\in\mathcal F}(gq)(x)=\bigvee_{q'\in g\mathcal F}q'(x)=\Bigl(\bigvee\mathcal F\Bigr)(x),
\]
because $g\mathcal F=\mathcal F$. The same computation applies to meets, and the top predicate is equivariant. The statement about subsets is the case $L=\{0,1\}$. Completeness at a fixed context is Proposition~\ref{prop:orbit-algebra} and its $L$-valued analogue, Proposition~\ref{thm:context-function-lattice} below.
\end{proof}

For the equality symmetry, Proposition~\ref{prop:internal-completeness} is \cite[Theorems~7.1 and~7.3]{AlexandruCiobanu2020}: $\Pfs(X)$ and $\Pred_{\mathrm{fs}}(X,L)$ are \emph{invariant complete lattices}, that is, every finitely supported subset has a least upper bound \cite[Definition~6.2]{AlexandruCiobanu2020}. They need not be complete lattices in the classical sense. For the equality symmetry, the finitely supported subsets of $\A$ are the finite and the cofinite ones \cite[Theorem~2.2]{AlexandruCiobanu2020}, so, for example when $\A$ is countably infinite in the metatheory, the family of singletons $\{a\}$ indexed by an infinite and coinfinite subset has no least upper bound in $\Pfs(\A)$ \cite[Proposition~7.1]{AlexandruCiobanu2020}; see also \cite[Remark~4.8]{KlinLelyk2019}. That family is not finitely supported, so supported operations applied to finitely supported data cannot produce it. Internal completeness is therefore the form of completeness that finitely supported semantics actually uses. The next observation is standard; it is the lattice-theoretic core of support arguments for fixed points such as \cite[Lemma~4.9]{KlinLelyk2019}.

\begin{lemma}[Fixed-point localization]\label{thm:fixed-point-localization}
Let a group $H$ act on a complete lattice $L$ by complete-lattice automorphisms, and let $L^H$ be the sublattice of $H$-fixed elements. Then $L^H$ is complete. If $F:L\to L$ is monotone and $H$-equivariant, then $F$ restricts to a monotone endomap of $L^H$, and \ $\lfp_L(F)=\lfp_{L^H}(F|_{L^H})$ and 
$\gfp_L(F)=\gfp_{L^H}(F|_{L^H})$.
In particular, the ambient least and greatest fixed points are $H$-fixed.
\end{lemma}

\begin{proof}
Arbitrary joins and meets of $H$-fixed elements are again $H$-fixed because every $h\in H$ preserves all joins and meets. Thus $L^H$ is complete. Equivariance of $F$ gives $F(L^H)\subseteq L^H$. If $p=\lfp_L(F)$ and $h\in H$, then $F(hp)=hF(p)=hp$,
so $hp$ is a fixed point of $F$. Since $h$ is an order automorphism, it permutes the fixed points and preserves their least element; hence $hp=p$. Therefore $p\in L^H$, and it is also the least fixed point of the restriction. The greatest-fixed-point argument is dual.
\end{proof}

Lemma~\ref{thm:fixed-point-localization} is used repeatedly below: whenever the transformer defining reachability, bisimilarity, a modal fixed point, or an abstract invariant is $G_{(S)}$-equivariant on the full powerset, its ambient least and greatest fixed points already lie in the context lattice. The same holds for lattice-valued predicates.

\begin{proposition}[Context function lattice]\label{thm:context-function-lattice}
For every complete lattice $L$ with trivial atom action,
\[
\Pred_S(X,L)\cong L^{G_{(S)}\backslash X}
\]
as complete lattices, with pointwise operations. Hence every monotone $S$-supported transformer $F:\Pred_{\mathrm{fs}}(X,L)\to\Pred_{\mathrm{fs}}(X,L)$ restricts to a monotone endomap of $\Pred_S(X,L)$ and has least and greatest fixed points there. If the transformer extends to a monotone $G_{(S)}$-equivariant endomap of the full function lattice $L^X$, the context fixed points coincide with the ambient ones.
\end{proposition}

\begin{proof}
An $S$-supported map $p:X\to L$ is constant on each $G_{(S)}$-orbit because~$L$ has trivial action. It therefore factors uniquely through the orbit projection $X\to G_{(S)}\backslash X$. Conversely, every map on the orbit set lifts to an $S$-supported map. Products of complete lattices are complete. Proposition~\ref{prop:criterion} shows that an $S$-supported transformer preserves $\Pred_S(X,L)$, and Knaster--Tarski gives fixed points in that complete lattice. Under the extension hypothesis, ambient coincidence follows from Lemma~\ref{thm:fixed-point-localization}.
\end{proof}

For $L=\{0,1\}$ this recovers Proposition~\ref{prop:orbit-algebra}. For $L=[0,1]$ it gives finitely supported fuzzy predicates; for cost, security, or information lattices it gives quantitative analyses. This extends the finitely supported fuzzy-set and lattice constructions of \cite{AlexandruCiobanu2018Fuzzy}, \cite[Section~7.2]{AlexandruCiobanu2020}.

The extension hypothesis in Proposition~\ref{thm:context-function-lattice} is needed only to compare with a given ambient transformer. For a monotone finitely supported transformer, least and greatest fixed points in the possibly incomplete lattice $\Pred_{\mathrm{fs}}(X,L)$ exist without an additional extension hypothesis.

The issue is not only to find a fixed point at the declared context, but to show that it remains least or greatest when compared with fixed points supported by other contexts. The following result establishes both assertions.

\begin{proposition}[Fixed points in the finitely supported function lattice]\label{prop:fs-fixed-points}
Let $L$ be a complete lattice with trivial atom action and let $F:\Pred_{\mathrm{fs}}(X,L)\to\Pred_{\mathrm{fs}}(X,L)$ be monotone and $S$-supported. Then $F$ has a least and a greatest fixed point in $\Pred_{\mathrm{fs}}(X,L)$, and they coincide with the least and greatest fixed points of the restriction of $F$ to $\Pred_S(X,L)$. Moreover, every monotone $G_{(S)}$-equivariant endomap of $L^X$ that agrees with $F$ on $\Pred_{\mathrm{fs}}(X,L)$ has the same least and greatest fixed points.
\end{proposition}

\begin{proof}
Define $\overline F:L^X\to L^X$ by $\overline F(p)=\bigvee\{F(q):\ q\in\Pred_{\mathrm{fs}}(X,L),\ q\leq p\}$, with pointwise joins. It is monotone, and it agrees with $F$ on $\Pred_{\mathrm{fs}}(X,L)$ because $F$ is monotone. For $g\in G_{(S)}$, the set $\Pred_{\mathrm{fs}}(X,L)$ is $g$-invariant, $g$ preserves pointwise joins, and $gF(q)=F(gq)$; hence $g\overline F(p)=\overline F(gp)$. Let $F'$ be any monotone $G_{(S)}$-equivariant endomap of $L^X$ agreeing with $F$ on $\Pred_{\mathrm{fs}}(X,L)$, for instance $F'=\overline F$. By Lemma~\ref{thm:fixed-point-localization}, $p_0=\lfp(F')$ is $G_{(S)}$-fixed, so $p_0\in\Pred_S(X,L)$ and $F(p_0)=F'(p_0)=p_0$. Every fixed point $q$ of $F$ in $\Pred_{\mathrm{fs}}(X,L)$ is a fixed point of $F'$, so $p_0\leq q$. Thus $p_0$ is the least fixed point of $F$ in $\Pred_{\mathrm{fs}}(X,L)$, and in particular also in $\Pred_S(X,L)$; it does not depend on the choice of $F'$. The greatest fixed point is treated dually, using $\widetilde F(p)=\bigwedge\{F(q):\ q\in\Pred_{\mathrm{fs}}(X,L),\ q\geq p\}$.
\end{proof}

Proposition~\ref{prop:fs-fixed-points} is the internal Tarski theorem for the lattices of Proposition~\ref{prop:internal-completeness}, together with the identification of internal, context, and ambient fixed points. For the equality symmetry, existence and support of these fixed points follow from the Strong Tarski theorem for invariant complete lattices \cite[Theorem~6.2 and Corollary~6.2]{AlexandruCiobanu2020} applied to \cite[Theorems~7.1 and~7.3]{AlexandruCiobanu2020}; that a Tarski theorem holds for lattices with joins of all finitely supported families is also noted in \cite[Remark~4.8]{KlinLelyk2019}. The argument of \cite{AlexandruCiobanu2020} uses that finitary permutations have finite order; for an arbitrary data symmetry that step is replaced by an application of $g^{-1}$, and the proof above avoids it altogether. In particular, the context fixed points in Propositions~\ref{thm:context-function-lattice} and~\ref{prop:ai-soundness} are the least fixed points among finitely supported predicates, without any extension hypothesis.

\subsection{Property III: uniform finiteness and terminating iteration}\label{sec:principle-finiteness}

Property~II gives fixed points without any finiteness hypothesis. The following notion from \cite{AlexandruCiobanu2020} governs when they are reached by finitely many iterations.

\begin{definition}[Uniform finiteness]\label{def:uniformly-finite}
A $G$-set $X$ is \emph{uniformly finite} if, for every finite $S\subseteq\A$, the set $X^{G_{(S)}}$ of $S$-supported elements of $X$ is finite.
\end{definition}

Equivalently, $X$ contains no infinite \emph{uniformly supported} subset, that is, no infinite subset whose elements share a common finite support; this is the form in which the condition is used in \cite[Theorems~5.7 and~5.9]{AlexandruCiobanu2020} and \cite{AlexandruCiobanu2019Uniform,AlexandruCiobanu2020Carpathian}. In the terminology of the finite-dependence profile $S\mapsto X^{G_{(S)}}$ of \cite{Ciobanu2026Profile}, $X$ is uniformly finite exactly when all its \emph{context levels} are finite. For the equality symmetry, orbit-finiteness implies finite context levels and a uniform bound on least-support sizes; under these hypotheses, the converse also holds \cite{Ciobanu2026Profile}. We use only finite context levels here. More generally, orbit-finiteness and uniform finiteness must be compared under explicit symmetry hypotheses; none of our termination proofs requires such a comparison.

\begin{proposition}[Uniformly finite predicate lattices]\label{prop:uniform-powerset}
Let $X$ be a nominal $G$-set and~$S$ a context. The $S$-supported elements of $\Pfs(X)$ form $\PS(X)$, and $|\PS(X)|=2^{m}$ when $G_{(S)}\backslash X$ has $m$ elements. Hence $\Pfs(X)$ is uniformly finite if and only if $G_{(S)}\backslash X$ is finite for every context $S$. The same holds for $\Pred_{\mathrm{fs}}(X,L)$ with $L$ finite and $|L|\geq2$, with $|\Pred_S(X,L)|=|L|^m$.
\end{proposition}

\begin{proof}
An element of $\Pfs(X)$ is $S$-supported exactly when it lies in $\PS(X)$, and $\PS(X)\cong\mathcal P(G_{(S)}\backslash X)$ by Proposition~\ref{prop:orbit-algebra}. Likewise $\Pred_S(X,L)\cong L^{G_{(S)}\backslash X}$ by Proposition~\ref{thm:context-function-lattice}. These sets are finite exactly when the orbit set is finite.
\end{proof}

For the equality symmetry, uniform finiteness is strictly weaker than orbit-finiteness. For example, $\Pfs(\A)$ consists of the finite and cofinite subsets of $\A$, and it has infinitely many $G$-orbits, one for each size of a finite set and one for each size of a finite complement; yet only $2^{|S|+1}$ subsets of $\A$ are supported by $S$, so $\Pfs(\A)$ is uniformly finite; the second constituent fails, since least supports of finite subsets are unbounded. The same holds for $\Pfs(\A^n)$ for $n\geq1$, since $G_{(S)}$ has finitely many orbits on $\A^n$ (Lemma~\ref{lem:local-oligomorphic} below); for the equality symmetry, this and related powersets of relations are treated in \cite{AlexandruCiobanu2021Relations}. By contrast, the set $\A^*$ of finite words is neither orbit-finite nor uniformly finite, because all words over~$S$ are $S$-supported.

The hypothesis of finitely many context orbits holds for every context as soon as the data symmetry is oligomorphic, that is, $G$ has finitely many orbits on $\A^n$ for every $n$; this is the case for the equality, ordered, and random-graph symmetries.

\begin{lemma}[Oligomorphic symmetries have finitely many context orbits]\label{lem:local-oligomorphic}
Let $G$ be oligomorphic and let $S$ be a context with $|S|=k$. Then $|G_{(S)}\backslash\A^n|\leq|G\backslash\A^{k+n}|$ for every $n$. Consequently $\A$ is uniformly finite, and for every $G$-invariant $X\subseteq\A^n$ the lattice $\Pfs(X)$ is uniformly finite.
\end{lemma}

\begin{proof}
Let $\bar s$ enumerate $S$ and send $\bar a\in\A^n$ to the $G$-orbit of $(\bar s,\bar a)\in\A^{k+n}$. If $(\bar s,\bar a)$ and $(\bar s,\bar b)$ lie in the same $G$-orbit, then $g(\bar s,\bar a)=(\bar s,\bar b)$ for some $g\in G$, so $g\in G_{(S)}$ and $g\bar a=\bar b$. Hence the map is injective on $G_{(S)}$-orbits, which gives the bound. The $S$-supported atoms are the singleton $G_{(S)}$-orbits on $\A$, so there are finitely many of them. A $G$-invariant $X\subseteq\A^n$ is a union of $G_{(S)}$-orbits of $\A^n$, so $G_{(S)}\backslash X$ is finite, and Proposition~\ref{prop:uniform-powerset} applies.
\end{proof}

This is the standard fact that pointwise stabilizers of finite sets in an oligomorphic group are again oligomorphic \cite{Cameron1990}; the condition that every context stabilizer has finitely many orbits is the local oligomorphicity of \cite{Ciobanu2026Profile}.

\begin{proposition}[Terminating fixed-point iteration]\label{thm:terminating-iteration}
Let $P$ be a $G$-set with a $G$-equivariant partial order. Suppose that $P$ has a least element $\bot$ and that the set $P^{G_{(S)}}$ is finite. Then every $S$-supported monotone map $f:P\to P$ has a least fixed point, which equals $f^n(\bot)$ for some $n<|P^{G_{(S)}}|$ and is supported by $S$. Dually, if $P$ has a greatest element $\top$, then $f$ has a greatest fixed point $f^n(\top)$ with $n<|P^{G_{(S)}}|$. In particular, this holds for every $S$ whenever $P$ is uniformly finite.
\end{proposition}

\begin{proof}
Every $g\in G$ is an order automorphism of $P$, so $g\bot$ is again a least element and hence $g\bot=\bot$. If $f^k(\bot)$ is $S$-supported, then $f^{k+1}(\bot)$ is $S$-supported by Proposition~\ref{prop:criterion}. Thus the ascending chain $\bot\leq f(\bot)\leq f^2(\bot)\leq\cdots$ lies in the finite set $P^{G_{(S)}}$, so $f^n(\bot)=f^{n+1}(\bot)$ for some $n<|P^{G_{(S)}}|$. If $q=f(q)$, then induction gives $f^k(\bot)\leq q$ for all $k$, so $f^n(\bot)$ is the least fixed point. The greatest fixed point is treated dually.
\end{proof}

For the equality symmetry and uniformly finite $P$, Proposition~\ref{thm:terminating-iteration} is \cite[Theorem~5.9 and Proposition~6.6]{AlexandruCiobanu2020}; see also \cite{AlexandruCiobanu2020Carpathian}. For context lattices it yields an explicit bound.

\begin{corollary}[Iteration bound for context lattices]\label{cor:iteration-bound}
Let $X$ be a nominal $G$-set and let $F:\Pfs(X)\to\Pfs(X)$ be monotone and supported by $S$. If $G_{(S)}\backslash X$ has $m<\infty$ elements, then
\[
\lfp(F)=F^m(\varnothing),\qquad \gfp(F)=F^m(X),
\]
both in $\PS(X)$, and they coincide with the fixed points of Proposition~\ref{prop:fs-fixed-points}. More generally, for $F:\Pred_{\mathrm{fs}}(X,L)\to\Pred_{\mathrm{fs}}(X,L)$ monotone and $S$-supported, with $L$ finite and every chain in $L$ having at most $h+1$ elements, the least and greatest fixed points are reached after at most $mh$ iterations from the bottom and top predicates.
\end{corollary}

\begin{proof}
The iterates from $\varnothing$ and from $X$ lie in $\PS(X)\cong\mathcal P(G_{(S)}\backslash X)$, in which a strictly increasing or strictly decreasing chain has at most $m+1$ elements. The fixed points so obtained are the least and greatest ones by the argument of Proposition~\ref{thm:terminating-iteration}, and the least and greatest fixed points are individually unique. In $\Pred_S(X,L)\cong L^m$ a strict chain has at most $mh+1$ elements.
\end{proof}

Corollary~\ref{cor:iteration-bound} does not require orbit-finiteness of $\Pfs(X)$. Its hypothesis does imply orbit-finiteness of $X$, because each $G$-orbit is a union of $G_{(S)}$-orbits. For example, for $X=\A$ with equality atoms and $|S|=k$, iteration stabilizes after at most $k+1$ steps although $\Pfs(\A)$ is orbit-infinite. It is a semantic statement. An algorithm needs, in addition, an effective presentation of the $G_{(S)}$-orbits and of $F$ on them, which must be supplied in addition to the finite cell presentations of Section~\ref{sec:structured}. Homogeneity alone does not provide algorithms for computing those presentations. The same counting of $S$-supported relations underlies the termination argument in the proof of \cite[Theorem~5.1]{KlinLelyk2019} for orbit-finite models; Corollary~\ref{cor:iteration-bound} isolates it for an arbitrary finitely supported transformer and an arbitrary data symmetry.

\section{Supported transition systems and determinization}\label{sec:automata}

This section adapts the subset construction and its support bounds from automata over data symmetries \cite{BojanczykKlinLasota2014} to automata that are finitely supported but not necessarily orbit-finite. Without orbit-finiteness the constructions exist but need not yield finite algorithms; the effective case is treated in Proposition~\ref{prop:effective-det}. In terms of Section~\ref{sec:principles}, the support bounds are instances of Property~I, and the effective criterion adds a finite orbit hypothesis in the spirit of Property~III.

\begin{definition}[FSS nondeterministic automaton]
A finitely supported nondeterministic automaton is a tuple
\[
\mathcal A=(Q,\Sigma,\Delta,I,F),
\]
where $Q$ and $\Sigma$ are nominal $G$-sets, $\Delta\subseteq Q\times\Sigma\times Q$ is a finitely supported transition relation, and $I,F\subseteq Q$ are finitely supported sets of initial and accepting states. The automaton is \emph{deterministic} if $I=\{q_0\}$ and $\Delta$ is the graph of a total function $\delta:Q\times\Sigma\to Q$; we then write $(Q,\Sigma,\delta,q_0,F)$. For $H\leq G$, a deterministic automaton is \emph{$H$-equivariant} if $\delta(hq,ha)=h\delta(q,a)$, $hq_0=q_0$, and $hF=F$ for all $h\in H$, $q\in Q$, and $a\in\Sigma$. For $H=G_{(S)}$ this means that $\delta$, $q_0$, and $F$ are $S$-supported. In the $H$-equivariant case $Q$ may, more generally, be a nominal $H$-set, as for the quotient automata of Section~\ref{sec:quotients}.
\end{definition}

The diagonal action extends to words $\Sigma^*$. Let $L(\mathcal A)\subseteq\Sigma^*$ be the recognized language.

\begin{proposition}[Support of the accepted language]\label{prop:language-support}
If $S_\Delta,S_I,S_F$ support $\Delta,I,F$, respectively, then
\[
S_\mathcal A=S_\Delta\cup S_I\cup S_F
\]
supports $L(\mathcal A)$.
\end{proposition}

\begin{proof}
For $g\in G_{(S_\mathcal A)}$, applying $g$ to every state and label of an accepting run yields an accepting run over the renamed word because $g$ preserves $\Delta$, $I$, and~$F$. Applying $g^{-1}$ gives the converse.
\end{proof}

\subsection{Supported-powerset determinization}

For $U\in\Pfs(Q)$ and $a\in\Sigma$, define
\[
\delta_{\mathcal A}(U,a)
 =\{q'\in Q:\exists q\in U\ (q,a,q')\in\Delta\},
\]
and put
\[
F^{\det}=\{U\in\Pfs(Q):U\cap F\neq\varnothing\}.
\]

\begin{proposition}[FSS determinization]\label{thm:determinization}
Every FSS nondeterministic automaton~$\mathcal A$ has a deterministic FSS automaton
\[
\det(\mathcal A)=
(\Pfs(Q),\Sigma,\delta_{\mathcal A},I,F^{\det})
\]
with initial state $I\in\Pfs(Q)$, recognizing the same language. The state space $\Pfs(Q)$ is a nominal $G$-set under direct image. The transition map is supported by every support of $\Delta$, and $F^{\det}$ is supported by every support of $F$.
\end{proposition}

\begin{proof}
First regard the displayed transition formula as a map on $\mathcal P(Q)\times\Sigma$. Let $S_\Delta$ support $\Delta$ and $g\in G_{(S_\Delta)}$. For $U\subseteq Q$ and $a\in\Sigma$,
\begin{align*}
q'\in\delta_{\mathcal A}(gU,ga)
&\Longleftrightarrow
\exists q\in U\ (gq,ga,q')\in\Delta\\
&\Longleftrightarrow
\exists q\in U\ (q,a,g^{-1}q')\in\Delta\\
&\Longleftrightarrow
g^{-1}q'\in\delta_{\mathcal A}(U,a).
\end{align*}
Hence $\delta_{\mathcal A}(gU,ga)=g\delta_{\mathcal A}(U,a)$. Proposition~\ref{prop:criterion} shows that the transition map is $S_\Delta$-supported and therefore maps supported inputs to supported subset states. Thus it restricts to the asserted map on $\Pfs(Q)\times\Sigma$. The same argument with intersections proves the assertion about $F^{\det}$.

The standard subset-construction induction shows that, after reading $w$, the deterministic state is exactly the set of states reachable in $\mathcal A$ after $w$. Acceptance is therefore preserved.
\end{proof}

\begin{corollary}[Support propagation along runs]\label{cor:reachable-support}
Let $w=a_1\cdots a_n$. Choose finite supports~$T_i$ of~$a_i$. The subset state reached from $I$ in $\det(\mathcal A)$ is supported~by
\[
S_I\cup S_\Delta\cup\bigcup_{i=1}^{n}T_i.
\]
\end{corollary}

\begin{proof}
Apply Proposition~\ref{prop:criterion} inductively along the deterministic run.
\end{proof}

In particular, if $S$ supports $I$, $\Delta$, and every input letter, the reached subset state lies in~$\mathcal P_S(Q)$.

\begin{remark}[Semantic closure versus orbit-finiteness]
Proposition~\ref{thm:determinization} does not assert that $\Pfs(Q)$ or the reachable determinized state space is orbit-finite. The construction separates two questions: determinization exists semantically in the supported universe, while finite symbolic representation requires an additional orbit or support bound. Indeed, nondeterministic register automata are more expressive than deterministic ones \cite{KaminskiFrancez1994}, so orbit-finite determinization fails in general. The same distinction appears for supported function spaces, which exist generally but need not be orbit-finite \cite{BojanczykNguyenStefanski2024}.
\end{remark}

\begin{proposition}[Effective local determinization criterion]\label{prop:effective-det}
Let $S$ support $\Delta$,~$I$, and $F$, and put $H=G_{(S)}$. Assume that:
\begin{enumerate}[label=(\roman*)]
\item the reachable part of $\det(\mathcal A)$ has finitely many effectively represented $H$-orbits;
\item the relevant $H$-orbits of reachable state--input pairs $(U,a)\in\Pfs(Q)\times\Sigma$ admit a finite effective enumeration, and the orbit of each successor $\delta_{\mathcal A}(U,a)$ is effectively computable;
\item acceptance is decidable on reachable state-orbit representatives.
\end{enumerate}
Then language emptiness and reachability of state orbits reduce to finite graph search on the \emph{local orbit graph}, the finite graph whose vertices are the reachable $H$-orbits of subset states and whose edges are defined in the proof.
\end{proposition}

\begin{proof}
The hypotheses define a finite quotient graph whose vertices are the reachable $H$-orbits of subset states. An edge from the orbit of $U$ to the orbit of $U'$ is present when some represented joint orbit contains a pair $(U,a)$ with $\delta_{\mathcal A}(U,a)=U'$. Every concrete run projects to a quotient path.

Conversely, consider a quotient path and suppose that the current concrete state lies in the orbit represented by the source of its next edge. If that edge is represented by $U\xrightarrow{a}U'$,
and the current state is $hU$, local equivariance gives
$hU\xrightarrow{ha}hU'$.
For a subsequent edge, first choose a group element carrying its selected source representative to the concrete state already reached, and apply the same argument. Induction therefore lifts every quotient path. Since the accepting subset is $H$-invariant, acceptance is constant on state orbits, and finite graph search is sound and complete for emptiness and reachable-orbit analysis.
\end{proof}

\section{Supported quotients and context-relative minimization}\label{sec:quotients}

A supported equivalence relation need not be preserved by all of $G$. Its quotient is naturally a system for the stabilizer of its support. Operationally, this means that an abstraction may depend on the finite interface fixed by the analysis. The statements of this section use Property~I.

\begin{proposition}[Local quotient theorem]\label{prop:local-quotient}
Let $R\subseteq X\times X$ be an equivalence relation supported by $S$. Then $G_{(S)}$ acts on $X/R$ by $g[x]_R=[gx]_R$.
If $f:X\to Y$ is an $S$-supported map compatible with equivalence relations $R$ on~$X$ and~$T$ on $Y$, both supported by $S$, then~$f$ induces a $G_{(S)}$-equivariant map $X/R\to Y/T$.
\end{proposition}

\begin{proof}
If $xRy$ and $g\in G_{(S)}$, then $(gx)R(gy)$ because $gR=R$. Thus the quotient action is well defined. If $X$ is nominal and $T$ supports $x$, then $T$ supports $[x]_R$ for the restricted group, since $(G_{(S)})_{(T)}=G_{(S\cup T)}$. Hence the quotient is a nominal $G_{(S)}$-set. Compatibility of~$f$ gives the induced map, and Proposition~\ref{prop:criterion} gives its equivariance under $G_{(S)}$.
\end{proof}

For a finitely supported language $L\subseteq\Sigma^*$, define the right Nerode equivalence by \ 
$u\Ner_L v \quad\Longleftrightarrow\quad \forall 
w\in\Sigma^*\ (uw\in L\Longleftrightarrow vw\in L)$.

\begin{proposition}[Supported Nerode quotient]\label{thm:nerode}
If $S$ supports $L$, then $S$ supports~$\Ner_L$. The quotient $\Sigma^*/\Ner_L$, with initial state $[\varepsilon]$, transitions $[u]\xrightarrow{a}[ua]$, and accepting states $\{[u]:u\in L\}$, is a deterministic $G_{(S)}$-equivariant automaton recognizing $L$. It is minimal in the sense of Corollary~\ref{cor:nerode-factorization}.
\end{proposition}

\begin{proof}
Let $g\in G_{(S)}$ and $u\Ner_L v$. For every $w$,
$(gu)w=g\bigl(u(g^{-1}w)\bigr)$. Since $gL=L$,
\[ (gu)w\in L \Longleftrightarrow u(g^{-1}w)\in L \Longleftrightarrow v(g^{-1}w)\in L \Longleftrightarrow (gv)w\in L.
\]
Hence $gu\Ner_L gv$. The transition $[u]\xrightarrow{a}[ua]$ and the accepting quotient states are therefore well defined and $G_{(S)}$-equivariant. That the quotient recognizes $L$ is the classical argument.
\end{proof}

For $S=\varnothing$ this is the Nerode construction for $G$-equivariant automata underlying the Myhill--Nerode theorem of \cite{BojanczykKlinLasota2014}; by Remark~\ref{rem:context-symmetry}, a general context gives the same construction for the data symmetry $(\A,G_{(S)})$.

\begin{corollary}[Canonical factorization]\label{cor:nerode-factorization}
Let $H=G_{(S)}$, and let $\mathcal D=(Q,\Sigma,\delta,q_0,F)$ be a deterministic $H$-equivariant automaton recognizing an $S$-supported language $L$, in which every state is reachable from $q_0$. Write $\delta^*(q,u)$ for the state reached from $q$ after reading $u$, and $[u]\cdot a=[ua]$. Then
\[
\varphi:Q\to\Sigma^*/{\Ner_L},\qquad \varphi(\delta^*(q_0,u))=[u]_{\Ner_L},
\]
is a well-defined surjective $H$-equivariant map with $\varphi(q_0)=[\varepsilon]$, $\varphi(\delta(q,a))=\varphi(q)\cdot a$, and $q\in F\iff\varphi(q)\in\{[u]:u\in L\}$. It is the unique map satisfying the first two equations.
\end{corollary}

\begin{proof}
Recall that $\delta^*(q,uw)=\delta^*(\delta^*(q,u),w)$.
\emph{Well-definedness.}  If $\delta^*(q_0,u)=\delta^*(q_0,v)$, then for every $w\in\Sigma^*$,
$uw\in L\iff\delta^*(\delta^*(q_0,u),w)\in F\iff\delta^*(\delta^*(q_0,v),w)\in F\iff vw\in L$,
so $u\Ner_L v$. Since every state is reachable, $\varphi$ is defined on all of $Q$.
\emph{Structure.}  Clearly $\varphi(q_0)=[\varepsilon]$. If $q=\delta^*(q_0,u)$, then $\delta(q,a)=\delta^*(q_0,ua)$, so $\varphi(\delta(q,a))=[ua]=\varphi(q)\cdot a$; moreover $q\in F$ if and only if $u\in L$. Surjectivity holds because $[u]=\varphi(\delta^*(q_0,u))$.
\emph{Equivariance.}  For $h\in H$, induction on the length of $u$, using $hq_0=q_0$ and $\delta(hq,ha)=h\delta(q,a)$, gives $\delta^*(q_0,hu)=h\delta^*(q_0,u)$. Hence $\varphi(hq)=\varphi(\delta^*(q_0,hu))=[hu]=h[u]=h\varphi(q)$, the action on classes being well defined by Propositions~\ref{prop:local-quotient} and~\ref{thm:nerode}.
\emph{Uniqueness.}  If $\psi(q_0)=[\varepsilon]$ and $\psi(\delta(q,a))=\psi(q)\cdot a$, induction on the length of $u$ gives $\psi(\delta^*(q_0,u))=[u]$, so $\psi=\varphi$ by reachability.
\end{proof}

\section{Support-aware operational semantics and concurrency}\label{sec:concurrency}

Finite-context dependency analysis applies not only to already constructed transition relations, but also to rule-generated operational semantics. This is important for process calculi and programming languages, where transitions are normally defined by inference rules rather than by an explicit relation. Rule-generated relations are handled by Property~I, and bisimilarity by Properties~II and~III.
Only finitary rules with positive premises are considered in this section. Formats with negative premises require an additional stratification or well-supported-proof semantics; the same support argument applies only after the corresponding derivation notion has been fixed.

Let $J$ be a nominal $G$-set of judgements. A finitary inference rule is a pair $(P,c)$ with finite premise set $P\subseteq J$ and conclusion $c\in J$. The action is componentwise, $g(P,c)=(gP,gc)$. For a rule set $\mathcal R$, write $\operatorname{Der}(\mathcal R)\subseteq J$ for the least set containing the conclusions of all rule instances whose premises already belong to it. The following is the fixed-context form of \cite[Theorem~3.6]{Pitts2006}, which is stated there for the equality symmetry; the proof carries over verbatim to an arbitrary data symmetry.

\begin{proposition}[Support of rule-generated semantics]\label{thm:sos-support}
If the inference system~$\mathcal R$ is supported by a finite context $S$, then the derivability set $\operatorname{Der}(\mathcal R)$ is supported by $S$. In particular, if $J$ consists of transition judgements $p\xrightarrow{\ell}q$, the least transition relation generated by an $S$-supported SOS specification is $S$-supported.
\end{proposition}

\begin{proof}
Let $g\in G_{(S)}$. Since $g\mathcal R=\mathcal R$, applying $g$ to a rule instance again yields a rule instance of $\mathcal R$. Induction on derivation height shows that $j\in\operatorname{Der}(\mathcal R)$ implies $gj\in\operatorname{Der}(\mathcal R)$. Applying~$g^{-1}$ yields equality.
\end{proof}

Proposition~\ref{thm:sos-support} gives a direct dependency analysis for structural operational semantics: names occurring only outside the support of the rule scheme cannot affect derivability. Standard operational specifications of the $\pi$-calculus and related name-passing calculi are globally equivariant; rule formats for nominal process calculi guarantee this property syntactically \cite{Milner1999,AcetoEtAl2019}. The next result shows that an equivariant operational semantics is closed under finitely supported policy restriction: intersecting its transition relation with an $S$-supported guard produces a semantics equivariant under the local group~$G_{(S)}$.

\begin{proposition}[Finite-context extension of an operational semantics]\label{prop:policy-extension}
Let $T_0\subseteq X\times\Lambda\times X$ be an equivariant transition relation and let $C\subseteq X\times\Lambda\times X$ be an $S$-supported guard, policy, or side condition. Then $T_S=T_0\cap C$ is an $S$-supported transition relation. Thus a globally equivariant process semantics may be refined by finitely many distinguished channels, principals, locations, or resources while retaining equivariance under $G_{(S)}$.
\end{proposition}

\begin{proof}
Both $T_0$ and $C$ are invariant under $G_{(S)}$, hence so is their intersection.
\end{proof}

Strong bisimilarity is the greatest fixed point of the transformer defined below, and is therefore characterized, and reasoned about, coinductively \cite{Sangiorgi2012}. For an $S$-supported labelled transition relation~$T$, define the strong-bisimulation transformer on binary relations by
\[
\mathcal B_T(R)=
\left\{(x,y):
\begin{array}{l}
\forall a\in\Lambda\ \forall x'\,
 \bigl(x\xrightarrow{a}x'\Rightarrow
       \exists y'\,(y\xrightarrow{a}y'\ \wedge\ (x',y')\in R)\bigr),\\[1mm]
\forall a\in\Lambda\ \forall y'\,
 \bigl(y\xrightarrow{a}y'\Rightarrow
       \exists x'\,(x\xrightarrow{a}x'\ \wedge\ (x',y')\in R)\bigr)
\end{array}
\right\}.
\]

\begin{proposition}[Contextual bisimulation]\label{thm:context-bisim}
The transformer $\mathcal B_T$ is monotone and $G_{(S)}$-equivariant on the full lattice $\mathcal P(X\times X)$. Consequently, ordinary strong bisimilarity is supported by $S$ and coincides with the greatest fixed point of $\mathcal B_T$ computed in $\PS(X\times X)$. Thus restricting the computation to the context lattice loses no bisimulation pairs. If $G_{(S)}\backslash(X\times X)$ has $m<\infty$ elements, the greatest-fixed-point iteration from $X\times X$ stabilizes after at most $m$ steps; if, in addition, the induced bisimulation transformer on the orbit powerset is effectively computable, bisimilarity is computable.
\end{proposition}

\begin{proof}
Monotonicity is immediate. Put $H=G_{(S)}$. For every relation $R$ and $g\in H$, transporting transitions and matching transitions by $g$ gives
\[
\mathcal B_T(gR)=g\mathcal B_T(R).
\]
So $\mathcal B_T$ is $H$-equivariant. Apply Lemma~\ref{thm:fixed-point-localization} to the complete lattice $\mathcal P(X\times X)$, on which~$H$ acts by complete-lattice automorphisms. Its greatest fixed point is the ordinary largest strong bisimulation, belongs to $\PS(X\times X)$, and equals the greatest fixed point of the restricted transformer. Since $\mathcal B_T$ maps $V$-supported relations to $(S\cup V)$-supported ones, it restricts to an $S$-supported monotone map on $\Pfs(X\times X)$, and Corollary~\ref{cor:iteration-bound} gives the bound on the number of iterations. Under the stated effectiveness hypothesis, the restricted lattice is the finite powerset of the pair-orbit set and its induced transformer is computable on orbit~cells.
\end{proof}

The statement concerns ordinary labels. A corresponding localization argument applies to a name-binding transition system once its residual object and binding-aware matching functional have been defined and shown to be $G_{(S)}$-equivariant. Established nominal transition-system semantics provide such residual and alpha-compatible constructions for several process-calculus bisimulations \cite{ParrowEtAl2021,AcetoEtAl2019}. We do not reproduce their binding-specific adequacy proofs here; the present statement adds only the finite-context localization step for predicates, policy guards, and behavioural~relations.

\section{Supported predicate-transformer logic}\label{sec:logic}

The preceding sections analyse operational transitions; the present section turns to specifications. When a transition relation and the atomic predicates of a program logic are supported by the same finite context $S$, the induced modal operators are monotone self-maps of the complete Boolean algebra $\mathcal P_S(X)$. Consequently, modal least and greatest fixed points are computed inside that context (Property~II), and under a finite orbit hypothesis within the iteration bound of Property~III.

Let $R\subseteq X\times X$ be supported by $S$. For $U,V\subseteq X$ define

\centerline{$\Post_R(U)=\{y:\exists x\in U\ (x,y)\in R\}$,}

\centerline{$\Pre_R(V)=\{x:\forall y\ ((x,y)\in R\Rightarrow y\in V)\}$.}

\begin{proposition}[Supported modal adjunction]\label{prop:modal-adjunction}
For all $g\in G_{(S)}$ and $U,V\subseteq X$,
\[
\Post_R(gU)=g\Post_R(U),\qquad \Pre_R(gV)=g\Pre_R(V).
\]
Consequently $\Post_R$ and $\Pre_R$ are $S$-supported maps on $\mathcal P(X)$, and they are $G$-equivariant when $R$ is. They map $\PS(X)$ into itself and $T$-supported subsets to $(S\cup T)$-supported subsets, and
$\Post_R(U)\subseteq V\iff U\subseteq\Pre_R(V)$. The same holds for the modalities obtained from the converse relation $R^{-1}$.
\end{proposition}

\begin{proof}
Let $g\in G_{(S)}$, so that $gR=R$. Then $y\in\Post_R(gU)$ if and only if $(gx,y)\in R$ for some $x\in U$, if and only if $(x,g^{-1}y)\in R$ for some $x\in U$, if and only if $g^{-1}y\in\Post_R(U)$. Similarly, writing $y=gy'$, we have $x\in\Pre_R(gV)$ if and only if every $y'$ with $(g^{-1}x,y')\in R$ lies in $V$, that is, if and only if $g^{-1}x\in\Pre_R(V)$. The support statements follow from Proposition~\ref{prop:criterion}, and the adjunction is the usual one between direct image and universal preimage along a relation. The converse relation $R^{-1}$ is again $S$-supported.
\end{proof}

Consider the positive modal $\mu$-calculus generated from $S$-supported atomic predicates by Boolean operations, the modalities induced by supported transition relations, and least and greatest fixed-point binders. Bound variables occur positively, as usual. 
Proposition~\ref{thm:mu-context} is the fixed-context form of \cite[Lemma~4.9]{KlinLelyk2019}, where the support of $\mu$-calculus denotations over equality and ordered atoms is bounded by the supports of the formula, the model, and the environment.

\begin{proposition}[Contextual fixed-point semantics]\label{thm:mu-context}
Let every transition relation and atomic valuation occurring in a modal fixed-point formula $\varphi$ be supported by a finite set $S$. For every environment assigning $S$-supported predicates to the free fixed-point variables, the denotation $\llbracket\varphi\rrbracket$ belongs to $\PS(X)$. Least and greatest fixed-point binders are interpreted by the corresponding fixed points in the complete lattice $\PS(X)$.
\end{proposition}

\begin{proof}
Proceed by structural induction. Atomic predicates and variable valuations lie in $\PS(X)$ by assumption. Boolean operations preserve $\PS(X)$ by Proposition~\ref{prop:orbit-algebra}, and the modal clauses preserve it by Proposition~\ref{prop:modal-adjunction}. A positive fixed-point body defines a monotone self-map of $\PS(X)$. The Knaster--Tarski theorem \cite{Tarski1955} supplies its least or greatest fixed point. Equivalently, Lemma~\ref{thm:fixed-point-localization} shows that the ambient fixed point is already $G_{(S)}$-invariant.
\end{proof}

\begin{corollary}[Finite local model checking]\label{cor:finite-model-checking}
If $G_{(S)}\backslash X$ is finite and the induced Boolean and modal operations on the orbit cells are effective, then model checking every formula of Proposition~\ref{thm:mu-context} reduces to ordinary finite fixed-point evaluation over the orbit-cell quotient.
\end{corollary}

\begin{proof}
By Proposition~\ref{prop:orbit-algebra}, $\PS(X)$ is the powerset of the finite orbit set. Boolean and modal operations act on this finite Boolean algebra, and positivity makes each fixed-point body monotone, and the usual finite $\mu$-calculus evaluation applies; each fixed-point iteration stabilizes within the bound of Corollary~\ref{cor:iteration-bound}.
\end{proof}

\begin{example}[One-register equality context]
Let $X=\A$ with equality atoms and let $S=\{r\}$. The stabilizer $G_{(S)}$ has two orbits on $X$: $\{r\}$ and $\A\setminus\{r\}$. Every unary $S$-supported predicate is therefore one of four predicates, and every supported modal fixed-point analysis of a unary transition system is evaluated over a four-element Boolean algebra. The infinite data domain is retained semantically, but the declared register context yields a finite abstract state space.
\end{example}

\section{Context-stratified abstract interpretation}\label{sec:ai}

Abstract interpretation relates concrete and abstract semantic domains by monotone maps, commonly a Galois connection, and computes invariants as fixed points \cite{CousotCousot1977}. Earlier FSS work translated invariant lattices, correctness relations, representation functions, Galois connections, and widening/narrowing to finitely supported structures \cite{AlexandruCiobanu2016Abstract}. The fixed-context formulation refines these constructions: every declared context carries a complete concrete predicate lattice and, under the hypotheses below, a complete abstract lattice. Thus fixed-point soundness is established locally even when the union of all supported strata is not~complete.

\subsection{Boolean and quantitative predicate domains}

The concrete domains of the analyses below are the context lattices $\Pred_S(X,L)$ and the finitely supported lattices $\Pred_{\mathrm{fs}}(X,L)$ of Section~\ref{sec:principle-completeness}. By Propositions~\ref{prop:internal-completeness} and~\ref{prop:fs-fixed-points} (Property~II), every monotone $S$-supported transformer has least and greatest fixed points among finitely supported predicates, and they lie in $\Pred_S(X,L)$. Property~III bounds the number of iterations when its finiteness hypotheses hold.

\begin{corollary}[Finite-orbit computation]\label{cor:finite-ai}
Suppose that $G_{(S)}\backslash X$ has $m<\infty$ orbits. Then $\Pred_S(X,L)\cong L^m$. If $L$ is finite and the transformer is effectively given on orbit representatives, its least and greatest fixed points are computable by finitely many iterations, bounded as in Corollary~\ref{cor:iteration-bound}. For Boolean predicates, the lattice has $2^m$ elements.
\end{corollary}

\begin{example}[Reachability]
Let $R\subseteq X\times X$ and $I\subseteq X$ be supported by~$S$. The transformer $\Phi(U)=I\cup\Post_R(U)$ is $S$-supported and monotone. Its least fixed point in $\PS(X)$ is the ordinary set of states reachable from $I$: by Lemma~\ref{thm:fixed-point-localization}, the ambient reachability fixed point is already $S$-supported. If $G_{(S)}\backslash X$ is finite, reachability is computed on the finite local orbit lattice.
\end{example}

\subsection{Supported Galois connections and approximations}

\begin{definition}[Finitely supported Galois connection]
Let $C$ and $A$ be $G$-sets. A pair of monotone maps
\[
\alpha:\Pfs(C)\rightleftarrows\Pfs(A):\gamma
\]
is an $S$-supported Galois connection if both maps are supported by $S$ and
\[
\alpha(U)\subseteq V
\quad\Longleftrightarrow\quad
U\subseteq\gamma(V)
\]
for all supported predicates $U,V$.
\end{definition}

\begin{proposition}[Restriction to context lattices]\label{thm:gc-restriction}
An $S$-supported Galois connection restricts to an ordinary Galois connection $\alpha_S:\PS(C)\rightleftarrows\PS(A):\gamma_S$
between complete Boolean algebras. If $G_{(S)}\backslash C$ and $G_{(S)}\backslash A$ are finite, it is represented by finite monotone maps between powersets of finite orbit sets.
\end{proposition}

\begin{proof}
Proposition~\ref{prop:criterion} shows that $\alpha$ and $\gamma$ preserve $S$-supported predicates. The adjunction remains valid after restriction. Proposition~\ref{prop:orbit-algebra} gives the finite representation.
\end{proof}

Let $F:\Pfs(C)\to\Pfs(C)$ be an $S$-supported monotone concrete transformer and put \ $F^\sharp=\alpha\circ F\circ\gamma$.

\begin{proposition}[Contextual abstract-interpretation soundness]\label{prop:ai-soundness}
The transformer $F^\sharp$ is $S$-supported and monotone, and
$\alpha\bigl(\lfp_{\PS(C)}(F)\bigr) \subseteq
\lfp_{\PS(A)}(F^\sharp)$.
If $F$ and $F^\sharp$ extend to monotone $G_{(S)}$-equivariant endomaps of ambient complete predicate lattices, the displayed context fixed points coincide with the corresponding ambient fixed points.
\end{proposition}

\begin{proof}
Support follows from Proposition~\ref{prop:criterion}, and monotonicity from that of~$\alpha,F,\gamma$. Let $a=\lfp(F^\sharp)$ in $\PS(A)$. Since $\alpha F\gamma(a)=F^\sharp(a)=a$,
adjunction gives $F\gamma(a)\subseteq\gamma(a)$. Hence $\lfp(F)\subseteq\gamma(a)$, and another use of adjunction yields $\alpha(\lfp(F))\subseteq a$. Under the stated extension hypothesis, the final statement follows from Lemma~\ref{thm:fixed-point-localization}.
\end{proof}

For the equality symmetry, the Pawlak approximations of a finitely supported equivalence relation form a finitely supported Galois connection on $\Pfs(X)$ \cite[Example~8.1]{AlexandruCiobanu2020}. The following statement records the same fact for an arbitrary data symmetry, together with its restriction to the context lattice.

\begin{proposition}[Supported rough abstraction]\label{prop:rough}
Let $E\subseteq X\times X$ be an equivalence relation supported by $S$. For $U\in\Pfs(X)$ define
\[
\overline E(U)=\{x:[x]_E\cap U\neq\varnothing\},
\qquad
\underline E(U)=\{x:[x]_E\subseteq U\}.
\]
Then $\overline E,\underline E:\Pfs(X)\to\Pfs(X)$ are monotone, supported by $S$, mapping $\PS(X)$ into itself,~and
\[
\overline E(U)\subseteq V
\quad\Longleftrightarrow\quad
U\subseteq\underline E(V).
\]
Thus $\overline E\dashv\underline E$ is an $S$-supported Galois connection, which restricts to the context lattice~$\PS(X)$.
\end{proposition}

\begin{proof}
The $G_{(S)}$-invariance of $E$ transports equivalence classes: $g[x]_E=[gx]_E$ for $g\in G_{(S)}$. Hence $\overline E(gU)=g\overline E(U)$ and $\underline E(gU)=g\underline E(U)$, so both approximations are $S$-supported maps; by Proposition~\ref{prop:criterion} they send $T$-supported subsets to $(S\cup T)$-supported subsets, and in particular preserve $\Pfs(X)$ and $\PS(X)$. For the adjunction, $\overline E(U)\subseteq V$ means that every $E$-class meeting $U$ is contained in $V$, which is equivalent to every point of $U$ having its $E$-class contained in $V$.
\end{proof}

Given a monotone transformer $F$ and an initial abstract element $a_0$, define the widening sequence by \ 
$a_{n+1}=a_n\nabla F(a_n)$.
Given a pre-fixed widening result $b_0$, meaning $F(b_0)\leq b_0$, define the narrowing sequence by \ $b_{n+1}=b_n\triangle F(b_n)$.

\begin{proposition}[Support-preserving widening and narrowing]\label{prop:widening}
Let $\nabla$ and $\triangle$ be widening and narrowing operators supported by $S$. If the initial widening element, the semantic transformer, and the initial narrowing element are supported by $S$, every element of the associated widening and narrowing sequences is supported by $S$.
\end{proposition}

\begin{proof}
Induct on the sequence length and apply Proposition~\ref{prop:criterion} to every use of the transformer and the approximation operator.
\end{proof}

The conclusion does not itself guarantee termination; that is supplied by the usual widening and narrowing hypotheses \cite{CousotCousot1979}. It guarantees that the analysis does not introduce additional support dependencies; an $S$-supported predicate may still contain states involving atoms outside $S$.

\section{Resource algebra and reversible components}\label{sec:resources}

Multisets represent consumable or repeatable resources, concurrent process populations, chemical reactants, and database bags. The FSS multiset theory of \cite{AlexandruCiobanu2015Multisets} extends finite-alphabet resource algebra to infinite alphabets while retaining finite algebraic support. FSS groups and their homomorphism theorems provide related algebraic models of reversible components \cite{AlexandruCiobanu2014Groups}, \cite[Sections~7.3--7.4]{AlexandruCiobanu2020}. The section uses Property~I throughout.

\subsection{Free finite-resource monoids}

Let $\Sigma$ be a nominal $G$-set and define
\[
\mathcal M_{\mathrm{fin}}(\Sigma)
 =\{m:\Sigma\to\mathbb N:
      \{a:m(a)\neq0\}\text{ is finite}\}.
\]
The action is $(gm)(a)=m(g^{-1}a)$ and addition is pointwise.

\begin{proposition}[Free commutative resource algebra]\label{prop:free-multiset}
$\mathcal M_{\mathrm{fin}}(\Sigma)$ is a commutative monoid in nominal $G$-sets. If $M$ is a commutative monoid in nominal $G$-sets, written additively, with equivariant addition and zero, and $h:\Sigma\to M$ is an $S$-supported map, there is a unique monoid homomorphism
\[
\widehat h:\mathcal M_{\mathrm{fin}}(\Sigma)\to M,
\qquad
\widehat h(m)=\sum_{a\in\Sigma}m(a)h(a),
\]
extending $h$, and $\widehat h$ is supported by $S$.
\end{proposition}

\begin{proof}
Every multiset has finite carrier. The union of chosen supports of its finitely many carrier elements supports the multiset, so $\mathcal M_{\mathrm{fin}}(\Sigma)$ is a nominal $G$-set. Pointwise addition and the zero multiset are equivariant, and the displayed sum is finite. If $g\in G_{(S)}$, then
\begin{align*}
\widehat h(gm)
&=\sum_a m(g^{-1}a)h(a)
 =\sum_b m(b)h(gb)\\
&=\sum_b m(b)g h(b)
 =g\widehat h(m).
\end{align*}
Thus $S$ supports $\widehat h$. The ordinary free-commutative-monoid proof gives uniqueness.
\end{proof}

The analogous construction with finite words gives the free noncommutative monoid $\Sigma^*$. The alphabet may be infinite and orbit-infinite. Each concrete word or resource multiset is nevertheless finitely generated, and interpretation of generators does not increase the external context support. In this sense the multiset constructions of \cite{AlexandruCiobanu2015Multisets} do not require orbit-finiteness of the alphabet.

\begin{example}[Context-dependent resource accounting]
Let $M$ be a commutative monoid with trivial action, for instance $(\mathbb N,+,0)$ for costs or a join-semilattice $(L,\vee,\bot)$ of security levels, and let $h(a)\in M$ assign a resource meaning to the name $a$. If $h$ depends only on a finite policy context $S$, then the accumulated meaning of every finite multiset is also $S$-supported. Renaming names outside the policy context commutes with evaluation.
\end{example}

\begin{proposition}[Supported multiset rewriting]\label{prop:multiset-rewriting}
Let $\mathcal R\subseteq\mathcal M_{\mathrm{fin}}(\Sigma)\times\mathcal M_{\mathrm{fin}}(\Sigma)$ be a rewrite-rule set supported by $S$. Define
\[
m\longrightarrow_{\mathcal R} n
\quad\Longleftrightarrow\quad
\exists (\ell,r)\in\mathcal R\ \exists c\;
   (m=\ell+c\ \text{and}\ n=r+c).
\]
Then $\longrightarrow_{\mathcal R}$ is an $S$-supported transition relation. If $T$ supports an initial multiset $m_0$, the reachable set from $m_0$ is supported by $S\cup T$. If the reachable multiset states have finitely many effectively represented $G_{(S\cup T)}$-orbits and the induced successor relation on orbit representatives is computable, reachability reduces to finite graph search.
\end{proposition}

\begin{proof}
For $g\in G_{(S)}$, applying $g$ to a rewriting witness $(\ell,r,c)$ gives the witness $(g\ell,gr,gc)$ because $g\mathcal R=\mathcal R$ and the monoid operation is equivariant. Thus the transition relation is $S$-supported. Reachability is the least fixed point of $U\mapsto\{m_0\}\cup\Post_{\longrightarrow_{\mathcal R}}(U)$ in the complete context lattice $\mathcal P_{S\cup T}(\mathcal M_{\mathrm{fin}}(\Sigma))$. The final assertion is the unlabeled instance of Proposition~\ref{prop:effective-det}: quotient by the effectively represented local orbits and compute the induced successor graph.
\end{proof}

This covers multiset-rewriting and Petri-net-style resource semantics over infinite alphabets, in the broader tradition of rewriting-based models of concurrency \cite{Meseguer1992}. It does not assert that the concrete reachable set is finite; the finite conclusion concerns the local orbit quotient and requires the stated effectiveness hypothesis.

\begin{proposition}[Local homomorphism theorem]\label{prop:local-hom}
Let $f:K\to L$ be an $S$-supported homomorphism between groups in nominal $G$-sets (groups whose operations are equivariant). Then $\ker f$ is an $S$-supported normal subgroup, $\im f$ is stable under $G_{(S)}$, and
$K/\ker f\cong\im f$ as $G_{(S)}$-groups.
\end{proposition}

\begin{proof}
For $g\in G_{(S)}$ and $x\in\ker f$, $f(gx)=g f(x)=g e=e$,
so $gx\in\ker f$. The image calculation is analogous. Proposition~\ref{prop:local-quotient} supplies the quotient action, and the ordinary first isomorphism theorem is $G_{(S)}$-equivariant.
\end{proof}

This is the finite-context form of the correspondence and isomorphism theorems established for nominal groups in \cite{AlexandruCiobanu2014Groups}. The same work proves closure of suitable wreath products, suggesting finite-context analogues of cascade decompositions for reversible data automata. Such decomposition theorems require additional orbit-finiteness or uniform-support hypotheses and remain a separate problem.

\section{Structured atoms and finite orbit-cell compilation}\label{sec:structured}

General data symmetries replace equality atoms by structures such as dense order, equivalence relations, or the random graph \cite{BojanczykKlinLasota2014}. The semantic FSS results above require only a group action and finite supports. Model-theoretic hypotheses enter only when a finite compilation is desired; this section describes the finite representations used to make Property~III effective. The general presentation below assumes a chosen orbit transversal; all finite compilation instances use only finitely many representatives. Their effective construction is an additional input to an algorithm. The following elementary representation is a special case of the treatment of nominal sets and equivariant maps in \cite[Sections~8--10]{BojanczykKlinLasota2014}; we state it in the form used for finite compilation.

\begin{proposition}[Orbit representation of supported maps]\label{thm:map-representation}
Let $H=G_{(S)}$ and let $X,Y$ be $G$-sets. After choosing one representative $x_O$ from each $H$-orbit $O\subseteq X$, the $S$-supported maps $f:X\to Y$ are in bijection with families
\[
(y_O)_{O\in H\backslash X},
\qquad
y_O\in Y^{H_{x_O}},
\]
where $H_{x_O}$ is the stabilizer of $x_O$. The correspondence is $y_O=f(x_O)$ and $f(hx_O)=hy_O$.
\end{proposition}

\begin{proof}
By Proposition~\ref{prop:criterion}, $S$-supported maps are precisely $H$-equivariant maps. If $k\in H_{x_O}$, then
$ky_O=kf(x_O)=f(kx_O)=f(x_O)=y_O$.
Conversely, define $f(hx_O)=hy_O$. If $hx_O=h'x_O$, then $h'^{-1}h\in H_{x_O}$ and therefore $hy_O=h'y_O$. The map is well defined and $H$-equivariant.
\end{proof}

\begin{corollary}[Finite cell presentation]\label{cor:finite-cell}
If $H\backslash X$ is finite and every fixed-point set $Y^{H_{x_O}}$ has an effective finite representation, then all $S$-supported maps $X\to Y$ have a finite cell-table representation. In particular, if $Y$ is finite with trivial action, an $S$-supported map is exactly an ordinary table on the finite orbit set $H\backslash X$.
\end{corollary}

The next statement is standard: expanding $\mathfrak A$ by constants for the elements of $S$ yields again an ultrahomogeneous structure in a finite signature, and such a structure has finitely many automorphism orbits on each finite power; see \cite[Chapter~7]{Hodges1993} and \cite{Macpherson2011}. We include the short argument for completeness.

\begin{proposition}[Homogeneous finite-relational atoms]\label{thm:homogeneous-cells}
Let $\mathfrak A$ be an ultrahomogeneous structure in a finite relational signature, with universe $\A$ and $G=\Aut(\mathfrak A)$. For every finite $S\subseteq\A$ and every $n$, the pointwise stabilizer $G_{(S)}$ has finitely many orbits on $\A^n$. Each orbit is defined over $S$ by the complete quantifier-free diagram of a representative tuple.
\end{proposition}

\begin{proof}
Because the relational signature and $S$ are finite and $n$ is fixed, only finitely many complete equality and relation patterns can occur on an $n$-tuple together with the named parameters from $S$. Tuples in the same $G_{(S)}$-orbit have the same complete quantifier-free diagram. Conversely, if two tuples have the same such diagram over $S$, the map fixing $S$ and sending one tuple to the other is an isomorphism between the induced finite substructures. Homogeneity extends it to an automorphism in $G_{(S)}$.
\end{proof}

\begin{corollary}[Finite compilation of context operations]\label{cor:homogeneous-compilation}
Under the hypotheses of Proposition~\ref{thm:homogeneous-cells}, every $S$-supported predicate on $\A^n$ is a finite union of quantifier-free cells over~$S$. Every $S$-supported operation from $\A^n$ to a finite control set with trivial action is a finite cell-indexed table. Consequently, whenever the state, label, or semantic carriers under consideration are finite disjoint unions of $G_{(S)}$-invariant subsets of finite powers of $\A$, the corresponding supported transitions, predicates, and finite-control semantic operations admit finite cell presentations.
\end{corollary}

This is a finite representability statement, not a uniform algorithm for arbitrary ultrahomogeneous structures. Effective compilation additionally requires procedures to enumerate realizable diagrams over the given parameters and to evaluate the specified predicates, transitions, and transformers on them. For the equality, rational-order, and random-graph examples, the cells themselves have the explicit descriptions below; semantic operations still require an effective specification.

For the equality and dense-order symmetries, the quantifier-free cell description is the support--definability correspondence of \cite[Theorem~23]{Ciobanu2026Synthese}: a relation on $\A^n$ is finitely supported exactly when it is first-order definable over a finite set of parameters, and then it is quantifier-free definable over any finite set supporting it. Proposition~\ref{thm:homogeneous-cells} extends the orbit-counting part of that argument to ultrahomogeneous structures in finite relational signatures.

Three examples make the cell structure explicit.

\begin{example}[Equality atoms]
For $|S|=k$, the $G_{(S)}$-orbits on $\A$ are the $k$ singleton parameter cells and the fresh cell $\A\setminus S$. Thus a unary supported transition into finite control is a table with $k+1$ input cells. On $\A^n$, the cells are equality patterns among the coordinates and the parameters.
\end{example}

\begin{example}[Dense order]
For $(\mathbb Q,<)$ and $S=\{s_1<\cdots<s_k\}$, the unary cells are the $k$ singleton points and the $k+1$ open intervals determined by them. Hence there are $2k+1$ unary cells. Higher-arity cells are finite order types over $S$. The same counts hold for the uncountable homogeneous order $(\mathbb R,<)$, since Proposition~\ref{thm:homogeneous-cells} makes no countability assumption.
\end{example}

\begin{example}[The countable random graph]
Let $\mathfrak R=(\A,E)$ be the countable random graph. Outside a context $S$ of size $k$, the unary orbit of a vertex is determined by its adjacency vector to $S$. Every vector is realized by the extension property, so there are $k+2^k$ unary cells, including the singleton parameter cells. Finite-control transition functions are therefore ordinary tables indexed by adjacency patterns. This gives a direct compilation of graph-register semantics into finite data.
\end{example}

\subsection{Case study: a finite-context authorization monitor}\label{sec:authorization-monitor}

We now carry out a complete finite-context dependency analysis. Let $G$ be the finite-permutation group on the equality atoms $\A$, fix distinct atoms $p,r\in\A$, and put $S=\{p,r\}$. The atoms $p$ and $r$ represent a trusted principal and a protected resource. Let
\[
\Sigma=\{\mathsf{login}(a):a\in\A\}\;\sqcup\;
       \{\mathsf{access}(a):a\in\A\},
\]
with $g\mathsf{login}(a)=\mathsf{login}(ga)$ and
$g\mathsf{access}(a)=\mathsf{access}(ga)$. The control set
\[
Q=\{q_0,q_1,q_{\mathrm{ok}},q_{\mathrm{bad}}\}
\]
has trivial action. Define a total deterministic transition map
$\delta:Q\times\Sigma\to Q$ as follows:
\begin{align*}
\delta(q_0,\mathsf{login}(a))&=
 \begin{cases}q_1,&a=p,\\ q_{\mathrm{bad}},&a\neq p,\end{cases}
&
\delta(q_1,\mathsf{access}(a))&=
 \begin{cases}q_{\mathrm{ok}},&a=r,\\ q_{\mathrm{bad}},&a\neq r,\end{cases}
\end{align*}
all other actions from $q_0$ or $q_1$ lead to $q_{\mathrm{bad}}$, and
$q_{\mathrm{ok}}$, $q_{\mathrm{bad}}$ are absorbing. The initial state is~$q_0$ and the sole accepting state is $q_{\mathrm{ok}}$.

\begin{proposition}[Dependency propagation and finite compilation]\label{thm:authorization-monitor}
For the authorization monitor above:
\begin{enumerate}[label=(\alph*)]
\item $S=\{p,r\}$ supports $\delta$, the accepted language $L_{p,r}$, and its right Nerode equivalence. Moreover, $S$ is the least support of both $\delta$ and $L_{p,r}$.
\item Relative to $G_{(S)}$, the atom set has the three unary cells

\centerline{$\{p\},\qquad \{r\},\qquad \A\setminus\{p,r\}$.}

Hence the tagged alphabet has six cells, and the infinite-alphabet monitor compiles to an ordinary deterministic automaton over the finite alphabet

\centerline{
$\{\mathsf{login},\mathsf{access}\}\times
\{p,r,\mathsf{fresh}\}$.}

\item The four control states are pairwise Nerode-distinct. Thus the displayed monitor is already the context-relative minimal deterministic automaton for $L_{p,r}$.
\end{enumerate}
\end{proposition}

\begin{proof}
Every $g\in G_{(S)}$ preserves the tests $a=p$ and $a=r$; since $Q$ has trivial action, Proposition~\ref{prop:criterion} gives
$\delta(gq,g\sigma)=g\delta(q,\sigma)=\delta(q,\sigma)$. Hence $S$ supports~$\delta$, and Proposition~\ref{prop:language-support} and Proposition~\ref{thm:nerode} propagate the same support bound to~$L_{p,r}$ and its Nerode equivalence.

To prove minimality of the context, let $T$ be a finite support of $L_{p,r}$ and suppose $p\notin T$. Choose $c\notin T\cup\{p,r\}$ and let $\pi$ exchange $p$ and $c$ while fixing every other atom. Then $\pi\in G_{(T)}$, but

\centerline{
$\mathsf{login}(p)\mathsf{access}(r)\in L_{p,r}, \qquad
\mathsf{login}(c)\mathsf{access}(r)\notin L_{p,r}$,}

\noindent
contradicting $T$-invariance. Thus $p\in T$; the same argument for $r$ shows $r\in T$. The proof for $\delta$ is identical, using its values on the corresponding~actions.

The cell statement is the equality-atom decomposition above, applied to the two tagged copies of $\A$. The transition table is constant on those six cells, so Corollary~\ref{cor:finite-cell} gives the finite compilation. Finally, all four states are reachable. The state $q_{\mathrm{ok}}$ is distinguished by the empty continuation; $q_{\mathrm{bad}}$ has no accepting continuation; $q_1$ accepts the continuation $\mathsf{access}(r)$ whereas $q_0$ does not; and $q_0$ accepts $\mathsf{login}(p)\mathsf{access}(r)$. Hence the four states are pairwise Nerode-distinct.
\end{proof}

\begin{proposition}[Why the finite context cannot be erased]\label{prop:authorization-no-global}
There is no globally $G$-equivariant transition map on the unchanged objects $Q$ and $\Sigma$ that realizes the authorization policy above.
\end{proposition}

\begin{proof}
Choose $c\notin\{p,r\}$ and let $\pi$ exchange $p$ and $c$ while fixing $r$. If $\widehat\delta:Q\times\Sigma\to Q$ were globally equivariant and agreed with the policy, then, because $Q$ has trivial action,
\[
\widehat\delta(q_0,\mathsf{login}(c))
 =\widehat\delta(\pi q_0,\pi\mathsf{login}(p))
 =\pi\widehat\delta(q_0,\mathsf{login}(p))
 =q_1.
\]
The policy requires the left-hand side to be $q_{\mathrm{bad}}$, a contradiction.
\end{proof}

Proposition~\ref{prop:authorization-no-global} concerns representation on fixed carriers; it is not an impossibility result for nominal techniques. In the slice presentation of Remark~\ref{rem:context-symmetry}, with $\bar s=(p,r)$, the monitor becomes the globally $G$-equivariant map
\[
(Q\times G\bar s)\times\Sigma\to Q\times G\bar s,\qquad
((q,(a,b)),\sigma)\mapsto(\delta_{a,b}(q,\sigma),(a,b)),
\]
where $\delta_{a,b}$ is the monitor for the policy pair $(a,b)$; alternatively, $p$ and $r$ may be named as constants. The fixed-context presentation keeps $Q$ unchanged and records $(p,r)$ once, as the support of $\delta$.

The six orbit cells are a canonical alphabet partition, not a minimal behavioural alphabet. For this monitor, the transition table depends only on the following three classes:
\[
\{\mathsf{login}(p)\},\qquad
\{\mathsf{access}(r)\},\qquad
\Sigma\setminus\{\mathsf{login}(p),\mathsf{access}(r)\}.
\]
Their transitions from $q_0$ are respectively $q_1,q_{\mathrm{bad}},q_{\mathrm{bad}}$, and from $q_1$ they are $q_{\mathrm{bad}},q_{\mathrm{ok}},q_{\mathrm{bad}}$; both absorbing states remain fixed. Thus the orbit compilation admits a further three-symbol behavioural compression. This does not change the four-state Nerode minimality or the least support $\{p,r\}$.

The case study also separates semantic and algorithmic conclusions. Finite support proves that the transition map, language, and Nerode relation remain within $S$; the equality-atom orbit decomposition independently supplies the six-symbol implementation. No general claim is made that every supported system is orbit-finite.

The case study uses the equality symmetry only for concreteness. Over the ordered symmetry ($\A=\mathbb Q$, $G=\Aut(\mathbb Q,<)$), the same monitor is again supported by $S=\{p,r\}$, but $G_{(S)}$ has five unary cells, namely $\{p\}$, $\{r\}$, and three open intervals, so the tagged alphabet has ten cells, and the transition table is constant across the interval cells. The choice of symmetry matters for order-dependent policies: a monitor that authorizes $\mathsf{access}(a)$ for every $a<r$ is $S$-supported for the ordered symmetry, whereas for the equality symmetry it has no finite support, because $\{a\in\mathbb Q:a<r\}$ is infinite and coinfinite.

The same monitor gives a compact fixed-point analysis. For the underlying transition relation, write
\[
\Pre^{\exists}(U)
 =\{q\in Q:\exists\sigma\in\Sigma\ \exists q'\in U\
                    (q\xrightarrow{\sigma}q')\}.
\]
The backward may-reachability transformer
\[
\Phi_{\mathrm{bad}}(U)=\{q_{\mathrm{bad}}\}\cup\Pre^{\exists}(U)
\]
is supported by $S$, because it is induced by the $S$-supported transition map. Its least fixed point is $\{q_0,q_1,q_{\mathrm{bad}}\}$; thus $q_{\mathrm{ok}}$ is the only control state from which a violation cannot be reached. The result has empty support because~$Q$ carries the trivial action. This illustrates that finite-context dependency analysis provides a sound upper bound and need not return a least support at every semantic stage.

The examples factor the analysis into two independent layers. Finite support establishes semantic closure and propagates the declared context through transitions, predicates, and quotient constructions. Finiteness of the orbit decomposition of $G_{(S)}$ on the relevant carriers (Proposition~\ref{thm:homogeneous-cells}) then converts those supported objects into finite cell representations.

\section{Summary of hypotheses}\label{sec:frontier}

Table~\ref{tab:frontier} summarizes, for each construction, the property of Section~\ref{sec:principles} that it uses, what finite support alone yields, and which further hypotheses are needed for finite or globally equivariant behaviour.

\begin{table}[htbp]
\centering
\caption{The property used by each construction, what finite support yields, and what further hypotheses add.}
\label{tab:frontier}
\footnotesize
\renewcommand{\arraystretch}{0.97}
\begin{tabularx}{\textwidth}{>{\raggedright\arraybackslash}p{.15\textwidth}>{\raggedright\arraybackslash}p{.085\textwidth}XX}
\toprule
Construction & Property & Supported conclusion & Additional condition for finite/global behaviour\\
\midrule
Data symmetry $(\A,G)$ & -- & The subgroup parameter specifies the structure programs may observe; finite supports then localize individual semantic objects. & Passing to a smaller group admits more predicates and operations, but may alter least-support and orbit-finiteness properties.\\
Fresh data & I & A deterministic $S$-supported enumeration outputs only elements fixed by~$G_{(S)}$. & For equality, order, and random-graph atoms these are exactly $S$; nondeterministic generation may select outside $S$.\\
Function space & I & All finitely supported maps form a nominal $G$-set. & The function space need not be orbit-finite \cite{BojanczykNguyenStefanski2024}.\\
Finite stores & I & Finite environments and register files carry explicit support bounds. & Canonical minimal contexts require a least-support property.\\
Finite-context policies & I & Policy parameters are recorded as supports of transitions and specifications without enlarging the control object. & A globally equivariant presentation on the unchanged control object may be impossible; a standard globally equivariant encoding requires contextualization or representation enrichment.\\
SOS semantics & I & An $S$-supported inference system generates an $S$-supported transition relation. & Finite-state verification requires a finite effective local quotient.\\
NFA determinization & I & The subset construction exists on $\Pfs(Q)$ and preserves language. & Orbit-finite determinization requires an orbit bound on reachable subsets.\\
Bisimulation and quotients & I, II, III & Contextual bisimilarity is a greatest fixed point; an $S$-supported congruence yields a quotient for $G_{(S)}$. & A global $G$-action requires equivariance; finite algorithms require effective orbit cells.\\
Predicate and fixed-point semantics & II, III & Finitely supported predicates are complete for finitely supported families; monotone supported transformers have fixed points agreeing with the context fixed points and those of compatible monotone stabilizer-equivariant ambient extensions. & Boolean termination: finitely many carrier orbits under $G_{(S)}$; quantitative termination also needs finite height. Algorithms require effective operations on cells.\\
Abstract interpretation & I, II & Supported Galois connections, widening, and narrowing preserve the declared context. & Termination and precision still require the usual abstract-domain hypotheses.\\
Words, multisets, and rewriting & I & Universal properties and multiset rewriting extend to arbitrary nominal alphabets. & Orbit-finite reachability or algebraic decomposition requires further restrictions.\\
\bottomrule
\end{tabularx}
\end{table}

Global equivariance is the empty-context case of finite-context supportedness, and finite-context supportedness does not imply orbit-finite representability. Conversely, orbit-finiteness of the carriers does not make a given transition, predicate, or semantic transformer finitely supported. Finite support is also not automatic: an ordinary set-theoretic construction may introduce an unsupported subset or choice function, and unrestricted choice, global orderings, and arbitrary products need not preserve finite support \cite{AlexandruCiobanu2015Multisets}, \cite[Theorem~3.1]{AlexandruCiobanu2020}. Each statement above therefore records its support hypotheses separately from the orbit hypotheses used for effective computation.

\section{Research directions}\label{sec:directions}

The three structural properties suggest the following research directions; each requires a more precise formulation before a novelty or openness claim can be made. We list them by property, together with the results of this paper on which they would build.

\paragraph{Property I: computing and bounding contexts}
\begin{enumerate}[label=(I.\arabic*)]
\item \emph{Support inference.}  The hierarchical construction of supports (Section~\ref{sec:fss-viewpoint}) is algorithmic in spirit. Develop a type system or proof-assistant tactic that infers, for each program phrase, a finite context supporting its denotation, with the bounds of Propositions~\ref{prop:criterion}, \ref{thm:term-support}, and~\ref{prop:finite-store} as typing rules, and determine when the inferred context is least for data symmetries with least supports.
\item \emph{Growth of contexts.}  Corollary~\ref{cor:reachable-support} and Proposition~\ref{prop:no-fresh-generation} bound support growth along deterministic evaluations by the supports of the inputs and operations. Quantify this growth for classes of systems (for instance, bounded-register or bounded-context systems) and relate it to the size of compiled representations in Section~\ref{sec:structured}.
\end{enumerate}

\paragraph{Property II: internal domain theory over data symmetries}
\begin{enumerate}[label=(II.\arabic*)]
\item \emph{Internally complete semantic domains.}  Proposition~\ref{prop:internal-completeness} covers powersets and lattice-valued predicates. Determine which other domains of program semantics (stream, trace, and failure domains, powerdomains, and domains of probabilistic predicates) are complete for finitely supported families over an arbitrary data symmetry, and whether internal and ambient fixed points coincide there as in Proposition~\ref{prop:fs-fixed-points}. For the equality symmetry, the Tarski, Bourbaki--Witt, and Tarski--Kantorovitch theorems of \cite[Chapters~5 and~6]{AlexandruCiobanu2020} are the natural starting point.
\item \emph{Choice-free selection.}  Where semantics or synthesis requires a selection that is not finitely supported, determine how much must be added to the context for a finitely supported selection to exist, and when the support-indexed description operator of \cite{Ciobanu2026PhilMath} suffices.
\end{enumerate}

\paragraph{Property III: finiteness and complexity}
\begin{enumerate}[label=(III.\arabic*)]
\item \emph{Beyond orbit-finiteness.}  Corollary~\ref{cor:iteration-bound} needs finitely many context orbits on the carrier, while the predicate lattice may be orbit-infinite, as $\Pfs(\A)$ illustrates. Characterize the nominal $G$-sets with finite context levels but unbounded least supports, in the sense of the finite-dependence profile \cite{Ciobanu2026Profile}, that arise in verification, and determine which algorithms of \cite{BojanczykKlinLasota2014,KlinLelyk2019} extend to them.
\item \emph{Complexity in the number of context orbits.}  The iteration bound $m$ of Corollary~\ref{cor:iteration-bound} is the number of $G_{(S)}$-orbits, which for equality atoms grows with the Bell numbers in the arity. Find symbolic representations of context lattices that avoid enumerating orbits, and establish lower bounds in terms of $|S|$ and the arity.
\item \emph{Infinite-height quantitative lattices.}  For $L=[0,1]$, Property~II gives fixed points but Property~III does not give termination. Under suitable continuity or contractivity hypotheses, determine convergence rates of iteration in $\Pred_S(X,[0,1])$ and whether widening operators can be chosen supported by the context of the analysis, extending Proposition~\ref{prop:ai-soundness} and the widening results of Section~\ref{sec:ai}.
\end{enumerate}

\section{Related work}

\paragraph{Nominal sets, FSS, and the finite support principle}
Nominal sets and their use for names, binding, and $\alpha$-equivalence are presented in \cite{Pitts2013}. The finite support principle, that uniquely specified constructions from finitely supported data using equivariant operations remain finitely supported, is stated in \cite[Theorem~3.5]{Pitts2006}; the $S$-finite support principle, together with the hierarchical construction of supports used in our proofs, is formulated within set theory with atoms in \cite[Section~1.3]{AlexandruCiobanu2020}. That monograph studies which classical theorems remain valid when all objects involved, over an arbitrary infinite set of atoms, are required to be finitely supported: choice principles and the generalized continuum hypothesis fail \cite[Theorem~3.1]{AlexandruCiobanu2020}, several notions of infinity separate \cite[Chapter~9]{AlexandruCiobanu2020}, and Tarski, Bourbaki--Witt, and Tarski--Kantorovitch fixed-point theorems hold for invariant complete lattices and posets, with specific fixed-point results for sets without infinite uniformly supported subsets \cite[Chapters~5 and~6]{AlexandruCiobanu2020}, \cite{AlexandruCiobanu2019Uniform,AlexandruCiobanu2020Carpathian}.

Forms of infinity and finiteness for finitely supported sets are compared in \cite{AlexandruCiobanu2022Infinity,AlexandruCiobanu2024Finite}, and for equality atoms, the finite-dependence profile of \cite{Ciobanu2026Profile} decomposes orbit-finiteness into finite context levels and a uniform bound on least supports. The symmetry group as a calibration parameter of the finitely supported universe, the support--definability correspondence, and the semantic--syntactic gap for uncountable atom sets are developed in \cite{Ciobanu2026Synthese}; the graded passage from global equivariance to invariance under the stabilizer of a finite context, with its Galois connection and the choice-free support-indexed description operator, is developed in \cite{Ciobanu2026PhilMath}. Properties~II and~III of Section~\ref{sec:principles} extend the FSS completeness and uniform-finiteness results to arbitrary data symmetries, and the FSS constructions used in Sections~\ref{sec:ai} and~\ref{sec:resources} originate in this line of work.

Most closure statements in the present paper are instances of the finite support principle, and Proposition~\ref{thm:sos-support} is the fixed-context form of the theorem on finitely supported inductive definitions \cite[Theorem~3.6]{Pitts2006}. The category whose morphisms are all finitely supported maps is studied in \cite{Crole2021}; for a fixed context, Remark~\ref{rem:context-symmetry} identifies supported maps on the given carriers with $G_{(S)}$-equivariant maps. Allowing arbitrary nominal $G_{(S)}$-carriers gives the stated equivalence with a slice of nominal $G$-sets.

\paragraph{Data symmetries and automata}
Data symmetries, nominal $G$-sets, automata over them, the corresponding Myhill--Nerode theorem, least supports, and finite representations of orbit-finite sets are developed in \cite{BojanczykKlinLasota2014}. Sections~\ref{sec:automata}, \ref{sec:quotients}, and~\ref{sec:structured} specialize parts of that theory to a fixed context, and Proposition~\ref{thm:map-representation} is a special case of \cite[Sections~8--10]{BojanczykKlinLasota2014}. The difference in emphasis is that the semantic statements here do not assume orbit-finiteness, which is reintroduced only for effective computation. Register automata \cite{KaminskiFrancez1994} are closely related to orbit-finite nominal automata over the equality symmetry \cite{BojanczykKlinLasota2014}; automata and logics on data words and trees are surveyed in \cite{Segoufin2006}; and symbolic automata \cite{VeanesBjorner2012} obtain finiteness from decidable alphabet theories rather than from symmetry. A programming language for hereditarily orbit-finite data is described in \cite{BojanczykEtAl2012}, and \cite{BojanczykNguyenStefanski2024} identifies when orbit-finiteness is preserved by function spaces. Supported sets \cite{Wissmann2023} give an alternative finite representation layer for nominal sets and automata.

\paragraph{Logics, operational semantics, and bisimulation}
For the $\mu$-calculus with atoms, \cite[Lemma~4.9]{KlinLelyk2019} bounds the support of denotations by the supports of the formula, the model, and the environment; model checking over orbit-finite models is decidable \cite[Theorem~5.1]{KlinLelyk2019}, and orbit-finite parity games are solved through finite orbit quotients \cite[Lemma~7.3]{KlinLelyk2019}. Proposition~\ref{thm:mu-context} and Corollary~\ref{cor:finite-model-checking} are fixed-context forms of these facts for an arbitrary data symmetry. Nominal transition systems and their modal logics, including binding actions and finitely supported infinite conjunctions, are treated in \cite{ParrowEtAl2021}. Rule formats for nominal process calculi \cite{AcetoEtAl2019} guarantee equivariance of an operational semantics syntactically; Propositions~\ref{thm:sos-support} and~\ref{prop:policy-extension} instead record the semantic consequence of an $S$-supported rule set or side condition. Rational fixed points of endofunctors on nominal sets \cite{MiliusSchroderWissmann2016} and the general theory of bisimulation and coinduction \cite{Sangiorgi2012} provide the coalgebraic and fixed-point background for Section~\ref{sec:concurrency}.

\paragraph{Algebra and program analysis}
Universal algebra over nominal sets, with equivariant operations, is studied in \cite{KurzPetrisan2010}. Within FSS, nominal groups and their homomorphism theorems \cite{AlexandruCiobanu2014Groups}, multisets over infinite alphabets \cite{AlexandruCiobanu2015Multisets}, fuzzy sets \cite{AlexandruCiobanu2018Fuzzy}, and abstract interpretation with invariant correctness relations, representation functions, Galois connections, and widening/narrowing \cite{AlexandruCiobanu2016Abstract} were developed earlier. Sections~\ref{sec:ai} and~\ref{sec:resources} restate these constructions for an arbitrary data symmetry and at a fixed context, where the relevant predicate lattices are complete in the classical sense, and Proposition~\ref{prop:fs-fixed-points} records when internal, context, and ambient fixed points coincide. Abstract interpretation itself follows \cite{CousotCousot1977,CousotCousot1979}, and rewriting-based models of concurrency follow \cite{Meseguer1992}.

\section{Conclusion}

The FSS theory of Alexandru and Ciobanu supplies the central mathematical viewpoint of this paper: semantic constructions must remain finitely supported, completeness is relative to supported families, and termination depends on the appropriate notion of finiteness. We develop these structural properties over arbitrary data symmetries and arbitrary infinite atom sets. Transfer gives explicit bounds on the contexts of composite constructions. Internal completeness yields fixed points among supported predicates and identifies them with context fixed points and those of compatible ambient extensions. Uniform finiteness gives finite convergence; for Boolean predicates the number of context orbits on the carrier bounds the iteration length, even when the predicate lattice is orbit-infinite.

Restriction to $G_{(S)}$ connects this FSS development with nominal semantics at a fixed context. Keeping $S$ explicit also exposes how dependencies propagate between constructions and separates semantic existence from the additional hypotheses required for finite, effective representations.

The applications show these properties at work: determinization and Nerode quotients without orbit-finiteness; supported rule systems and bisimilarity with explicit iteration bounds; modal, quantitative, and abstract-interpretation semantics at a fixed context; resource algebras over orbit-infinite alphabets; and finite cell tables for homogeneous structures, illustrated by an authorization monitor whose two-atom policy context propagates to its language, its Nerode equivalence, and its six-letter compiled alphabet. Section~\ref{sec:directions} lists problems for each property; among them, support inference for program phrases, internal domain theory over data symmetries, and algorithms for sets with finite context levels but unbounded supports appear to be the most directly useful for verification.

\section*{Acknowledgements}
The author thanks Andrei Alexandru, his former PhD student, for his contributions to the development of the theory of finitely supported structures.

\end{document}